\documentclass[journal,12pt,onecolumn,draftclsnofoot,]{IEEEtran}

\usepackage[T1]{fontenc}
\usepackage{color}
\usepackage{varioref}
\usepackage{textcomp}
\usepackage{amsthm}
\usepackage{amsmath}
\usepackage{amssymb}
\usepackage{graphicx}
\usepackage{setspace}
\usepackage{esint}
\usepackage{algorithm}
\usepackage{algpseudocode}
\usepackage{pifont}
\usepackage{wrapfig}
\usepackage{multirow}
\usepackage{tabu}
\usepackage{amsmath}

\usepackage{enumitem}
\usepackage{cite}
\usepackage{cleveref}
\makeatletter

\newcommand{\Rmnum}[1]{\expandafter\@slowromancap\romannumeral #1@}
\makeatother

\newtheorem{theorem}{Theorem}
\newtheorem{lemma}{Lemma}
\newtheorem{remark}{Remark}

\newtheorem{corollary}{Corollary}

\newtheorem{thm}{\protect\theoremname}
\newtheorem{prop}[thm]{\protect\propositionname}
\providecommand{\propositionname}{Proposition}

\ifCLASSINFOpdf

\else

\fi

\usepackage[T1]{fontenc}
\usepackage{etoolbox}

\makeatletter
\patchcmd{\maketitle}{\@fnsymbol}{\@alph}{}{}  
\makeatother

\title{On the Capacity Region of a Cache-Aided Gaussian Broadcast Channel with Multi-Layer Messages}
\author{\IEEEauthorblockN{Mohammad Mohammadi Amiri and
Deniz G\"und\"uz}

}
\date{}

\begin{document}

\maketitle


\begin{abstract}
A cache-aided $K$-user Gaussian broadcast channel (BC) is studied. The transmitter has a library of $N$ files, from which each user requests one. The users are equipped with caches of different sizes, which are filled without the knowledge of the user requests in a \textit{centralized} manner. Differently from the literature, it is assumed that each file can be delivered to different users at different rates, which may correspond to different quality representations of the underlying content, e.g., scalable coded video segments. Accordingly, instead of a single achievable rate, the system performance is characterized by a rate tuple, which corresponds to the vector of rates users' requests can be delivered at. The goal is to characterize the set of all achievable rate tuples for a given total cache capacity by designing joint cache and channel coding schemes together with cache allocation across users. Assuming that the users are ordered in increasing channel quality, each file is coded into $K$ layers, and only the first $k$ layers of the requested file are delivered to user $k$, $k= 1, \ldots, K$. Three different coding schemes are proposed, which differ in the way they deliver the coded contents over the BC; in particular, \textit{time-division}, \textit{superposition}, and \textit{dirty paper coding} schemes are studied. Corresponding achievable rate regions are characterized, and compared with a novel outer bound. To the best of our knowledge, this is the first work studying the delivery of files at different rates over a cache-aided noisy BC.      
\end{abstract}


\section{Introduction}\label{Intro}
\makeatletter{\renewcommand*{\@makefnmark}{}
\footnotetext{Part of this work was presented at the IEEE International Symposium on Information Theory, Colorado, USA, June 2018 [10].}\makeatother}In the \textit{coded caching} framework introduced in \cite{MaddahAliCentralized}, transmission is performed over two phases: in the \textit{placement phase}, which takes place during off-peak hours, users fill their caches without knowing the particular demands. Once the demands are revealed, they are satisfied simultaneously over the \textit{delivery phase}. Here, we consider a Gaussian broadcast channel (BC) from the server to the users during the delivery phase. Cache-aided Gaussian BC is studied in \cite{HuangFadingChannelcodedcaching} with and without fading, and in \cite{PooyaCaireKhalajPhysicalLayerJournal,PetrosEliaTopological} focusing on the high SNR regime. A packet-erasure BC is considered in \cite{ShirinErasureChannelJournal} and \cite{MohammadDenizPacketErasureJournal}. A degraded BC is considered in \cite{ShirinWiggerYenerCacheAssingment}, where the placement phase is performed in a centralized manner with the full knowledge of the channel during the delivery phase. In \cite{MohammadDenizJSACPower} delivery over a Gaussian BC is studied from an energy efficiency perspective, assuming that the channel conditions in the delivery phase are not known during the placement phase, for both centralized and decentralized caching scenarios.

In most of the existing literature on coded caching, the key assumption is that the files in the library are coded at a single common rate, and each user requests one file from the library in its entirety. Accordingly, the objective function in \cite{ShirinErasureChannelJournal,MohammadDenizPacketErasureJournal,ShirinWiggerYenerCacheAssingment} is to maximize the common rate of the messages that can be delivered to all the users, and the supremum of the achievable rates is defined as the \textit{capacity} of the caching network. In \cite{QianqianDenizDifDisReq}, the authors relaxed this assumption and allowed each user to request the files at a different quality, and equivalently, at a different rate. However, the required rates at which the contents must be delivered are assumed to be given as part of the problem definition in \cite{QianqianDenizDifDisReq}, and the goal is to find the minimum number of bits that must be delivered over an error-free shared delivery channel \cite{QianqianDenizDifDisReq}. In this work, similarly to \cite{QianqianDenizDifDisReq}, we allow the users to request the files at different rates; however, differently from \cite{QianqianDenizDifDisReq}, considering a Gaussian BC in the delivery phase, we aim at characterizing the rate tuples at which the requested contents can be delivered to the users \cite{MohammadDenizMutltiLayerISIT18}.

We argue that this formulation allows us to better exploit the asymmetric resources available to users for content delivery over a noisy BC. To see the difference between the scalar capacity definition used in \cite{ShirinWiggerYenerCacheAssingment} and the capacity region formulation proposed here, consider a Gaussian BC without any caches, i.e., $M=0$. In this case, the capacity as defined in \cite{ShirinWiggerYenerCacheAssingment} is limited by the rate that can be delivered to the worst user, whereas with our formulation any rate tuple within the capacity region of the underlying BC is achievable, providing a much richer characterization of the performance for cache-aided delivery over a noisy BC.

The motivation here is to deliver the contents at higher rates to users with better channels, rather than being limited by the weak users. As proposed in \cite{QianqianDenizDifDisReq}, the multiple rates of the same file may correspond to the video files in the library encoded into multiple quality layers using scalable coding, so the user with a higher delivery rate receives a better quality description of the same file. Accordingly, each file in the library is coded into $K$ layers, $K$ being the number of users, ordered in increasing channel qualities, where user $k$ receives layers 1 to $k$ of its request, $k=1, \ldots, K$. We consider a centralized placement phase, and assume that the channel qualities of the users in the delivery phase are known in advance. By allowing users to have different cache capacities (similarly to \cite{MohammadQianqianDenizDifCacheCap} considering an error-free shared link during the delivery phase), we consider a total cache capacity in the network as a constraint, and optimize cache allocation across the users and different layers of the files. Contents cached during the placement phase provide multicasting opportunities to the server to deliver the missing parts in the same layer of the files to different users. When delivering these coded contents to users over the underlying BC, we consider three different techniques. Corresponding coding schemes are called  joint cache and time-division coding (CTDC), joint cache and superposition coding (CSC), and joint cache and dirty paper coding (CDPC). We also present an outer bound on the rate region when the placement phase is constrained to uncoded caching, and compare it with the achievable rate tuples obtained though the proposed coding schemes.

\textit{Notations:} $\mathbb{R}$ and $\mathbb{R}^{++}$ represent sets of real values and positive real values, respectively. For any arbitrary non-empty set $\cal G$ with cardinality $\left| \mathcal G \right|$, we denote the $\binom{\left| \mathcal G \right|}{i}$ $i$-element subsets of $\cal G$ by $\mathcal{S}^i_{\mathcal{G},1}, \dots, \mathcal{S}^i_{\mathcal{G},\binom{\left| \mathcal G \right|}{i}}$, for $i=1, ..., \left| \mathcal G \right|$. For $g \notin \mathcal G$, we define $\left\{ \mathcal G,g \right\}  \buildrel \Delta \over =  \mathcal{G} \bigcup \left\{ g \right\}$, and for $\mathcal{H} \subset \mathcal{G}$, $\mathcal{G} \backslash \mathcal{H}$ represents $\left\{ j:j \in \mathcal{G}, j \notin \mathcal{H} \right\}$. For two integers $i$ and $j$, $j \ge i$, $[i:j]$ denotes the set $\{ i,i+1, ..., j \}$. For any positive real number $q$, we define $[q] \triangleq \{1, \ldots, \lceil q\rceil\}$. We define, for two real values $p \ge 0$ and $q > 0$, $C_q^p \buildrel \Delta \over = \frac{1}{2} {\log _2}\left( 1 + p/q \right)$, and $\bar{p} \buildrel \Delta \over = 1-p$. Notation $\bar \oplus$ represents bitwise XOR operation where the arguments are first zero-padded to have the same length as the longest argument. $\mathcal{N} \left( 0,a^2 \right)$ denotes a zero-mean normal distribution with variance $a^2$.

\section{System Model and Preliminaries}\label{SystemModel}

We consider cache-aided content delivery over a $K$-user Gaussian BC. The transmitter has a library of $N$ files, $\mathbf{W} \buildrel \Delta \over = W_1, ..., W_N$. File $W_j$ is coded into $K$ layers $W_{j}^{(1)}, \dots, W_{j}^{(K)}$, such that layer $W_{j}^{(l)}$ is distributed uniformly over the set $\left[ \left\lceil 2^{nR^{(l)}} \right\rceil \right]$, where $R^{(l)}$ represents the rate of the $l$-th layer and $n$ denotes the blocklength, for $j=1, ..., N$, and $l=1, ..., K$. We denote the $l$-th layers of all the files by $\textbf{W}^{(l)} \buildrel \Delta \over = W_1^{(l)}, ..., W_N^{(l)}$, for $l \in [K]$.

Assume that user $k$, $k \in [K]$, has a cache of capacity $nM_k$ bits, which is filled during the \textit{placement phase} without the knowledge of the user demands. 
User demands are revealed and satisfied simultaneously in the \textit{delivery phase}. Each user requests a single file from the library, where $W_{d_k}$, $d_k \in [N]$, denotes the file requested by user $k \in [K]$. For a demand vector $\textbf{d} \buildrel \Delta \over =(d_1, ..., d_K)$, the users are served by a common message $X^n (\textbf{W}) \buildrel \Delta \over = \left( X_1 (\textbf{W}), \dots, X_n (\textbf{W}) \right)$ satisfying the average power constraint. User $k$, $k \in [K]$, receives $Y^n_k (\textbf{W}) \buildrel \Delta \over = \left( Y_{k,1} (\textbf{W}), \dots, Y_{k,n} (\textbf{W}) \right)$ through a Gaussian channel 
\begin{equation}\label{ChannelModel} 
{Y^n_{k}(\textbf{W})} = {X^n(\textbf{W})} + {Z^n_{k}},
\end{equation}
where $Z^n_k \buildrel \Delta \over = \left( Z_{k,1}, \dots, Z_{k,n} \right)$, and $Z_{k,i}$ is the independent zero-mean real Gaussian noise with variance $\sigma_k^2$ at user $k$ at the $i$-th channel use. Without loss of generality we order the users in increasing channel quality, i.e., we assume that $\sigma _1^2 \ge \sigma _2^2 \ge \cdots \ge \sigma _K^2$. We define $\boldsymbol{\sigma} \buildrel \Delta \over = \left( \sigma_1, \dots, \sigma_K \right)$.

Placement phase is performed in a centralized manner assuming $\boldsymbol{\sigma}$ is known. An $\left( n, R^{(1)}, \ldots, \right.$ $\left. R^{(K)}, M_1, \ldots, M_K \right)$ code consists of the following:
\begin{itemize}
\item $K$ caching functions ${\phi _{k}}$, $k \in [K]$, where
\begin{equation}\label{CachingFunction}
    {\phi _{k}}: {\left\{ {\left[ {\left\lceil {{2^{n{R^{(1)}}}}} \right\rceil } \right] \times \cdots \times \left[ {\left\lceil {{2^{n{R^{(K)}}}}} \right\rceil } \right]} \right\}^N} \times {\mathbb{R}^{++}}^{K} \to \left[ \left\lfloor 2^{nM_k} \right\rfloor \right]
\end{equation}
maps $\textbf{W}$ and $\boldsymbol{\sigma}$ to the cache content $U_k$ of user $k$, i.e., $U_k = {\phi _{k}} \left( \textbf{W}, \boldsymbol{\sigma} \right)$. 
\item An encoding function
\begin{equation}\label{EncodingFunction}
    \psi: {\left\{ {\left[ {\left\lceil {{2^{n{R^{(1)}}}}} \right\rceil } \right] \times \cdots \times \left[ {\left\lceil {{2^{n{R^{(K)}}}}} \right\rceil } \right]} \right\}^N} \times {\mathbb{R}^{++}}^{K} \times \left[N\right]^K \to \mathbb{R}^n,
\end{equation} 
which generates the channel input as $X^n (\textbf{W}) = \psi \left( \textbf{W}, \boldsymbol{\sigma},\textbf{d} \right)$, for demand vector $\textbf{d}$, satisfying the average power constraint $\frac{1}{n}\sum\nolimits_{i = 1}^n {X_i^2(\textbf{W})} \le P$.
\item $K$ decoding functions ${\mu _{k}}$, $k \in [K]$, where, for a demand vector $\textbf{d}$,
\begin{align}\label{DecodingFunction}
    {\mu _{k}}: \mathbb{R}^n &\times \left[ \left\lfloor 2^{nM_k} \right\rfloor \right] \times \left[ N \right]^K \to {\left[ {\left\lceil {{2^{n{R^{(1)}}}}} \right\rceil } \right] \times \cdots \times \left[ {\left\lceil {{2^{n{R^{(k)}}}}} \right\rceil } \right]}
\end{align}
reconstructs the layers ${\hat W}_{d_k}^{(1)}, \dots, {\hat W}_{d_k}^{(k)}$ from the channel output $Y_k^n \left( \textbf{W} \right)$ and cache content $U_k$. 
\end{itemize}

The probability of error is defined as
${P_{e}} \buildrel \Delta \over = \Pr \left\{ \bigcup\nolimits_{\textbf{d} \in {[N]}^K}\bigcup\nolimits_{k = 1}^K \bigcup\nolimits_{l = 1}^k {\left\{ {{{\hat W}_{{d_k}}^{(l)}} \ne {W_{{d_k}}^{(l)}}} \right\}}  \right\}$.

Note that the generated code implicitly assumes that user $k$ is interested only in the first $k$ layers of its demand, i.e., $W_{d_k}^{(1)}, \dots, W_{d_k}^{(k)}$, for $k \in [K]$. In a more general formulation, we could instead consider an arbitrary ordering of the rates among the users, but here the goal is to deliver a higher rate to a user with a better channel.

For a given total cache capacity $M$, we say that the rate tuple $\left( R_1, \dots, R_K \right)$ is achievable if for every $\varepsilon > 0$, there exists an $\left( n, R^{(1)}, \dots,R^{(K)}, M_1, \dots, M_K \right)$ code, which satisfies ${P_{e}} < \varepsilon$, $R_k \le \sum\nolimits_{l=1}^{k} R^{(l)}$, and $\sum\nolimits_{k=1}^{K} M_k \le M$. For average power constraint $P$ and a total cache capacity $M$, the capacity region $\mathcal{C} (P,M)$ of the caching system described above is defined as the closure of the all achievable rate tuples. Our goal is to find inner and outer bounds on $\mathcal{C} (P,M)$.



Next, we present some definitions that will simplify our ensuing presentation. For a fixed value of $t$, $t \in [K-1]$, we define $g_l \buildrel \Delta \over = \sum\nolimits_{j=1}^{l} \binom{K-j}{t}$, $\forall l \in [K-t]$, and let $g_0 =0$. We note that $g_{K-t} = \binom{K}{t+1}$. We denote the set of users $[l:K]$ by $\mathcal{K}_l$, for $l \in [K]$. We label $(t+1)$-element subsets of users in $\mathcal{K}_1$, so that the subsets with the smallest element $l$ are labelled as
\begin{equation}\label{LabelTheSubfilesK1}
    {\mathcal{S}_{\mathcal{K}_1, 1+ g_{l-1}}^{t+1}}, \dots, {\mathcal{S}_{\mathcal{K}_1,g_l}^{t+1}}, \quad \mbox{for $l=1, ..., K-t$}. 
\end{equation}
Thus, we have, for $l \in [K-t]$, 
\begin{align}\label{FamilyLabelTheSubfilesK1Kk}
    &\left\{ {\mathcal{S}_{\mathcal{K}_{1}, 1+ g_{l-1}}^{t+1}} \backslash \{ l \}, \dots, {\mathcal{S}_{\mathcal{K}_1,g_l}^{t+1}} \backslash \{ l \} \right\} = \left\{ {\mathcal{S}_{\mathcal{K}_{l+1}, 1}^{t}}, \dots, {\mathcal{S}_{\mathcal{K}_{l+1},\binom{K-l}{t}}^{t}}  \right\},
\end{align}
i.e., the family of all $(t+1)$-element subsets of $\mathcal{K}_1$ excluding $l$, which is their smallest element, is the same as the family of all $t$-element subsets of $\mathcal{K}_{l+1}$. We note that the number of subsets of users in both sets in \eqref{FamilyLabelTheSubfilesK1Kk} is $\binom{K-l}{t}$, $l \in [K-t]$. Without loss of generality, we label the subsets of users so that, for $l \in [K-t]$,
\begin{equation}\label{LabelSubsetsInsideeachGroupK1Kk}
    {\mathcal{S}_{\mathcal{K}_{1}, i+ g_{l-1}}^{t+1}} \backslash \{ l \} = {\mathcal{S}_{\mathcal{K}_{l+1}, i}^{t}}, \quad \mbox{for $i \in \left[ \binom{K-l}{t} \right]$}. 
\end{equation}

\section{Achievable Schemes}\label{Results}
Here we present three different inner bounds on $\mathcal{C} (P,M)$. 

\subsection{Joint Cache and Time-Division Coding (CTDC)}\label{CTDCMainResults}
In the following, we present an achievable rate region achieved by the CTDC scheme. With CTDC, the missing bits corresponding to the layers in $\textbf{W}^{(l)}$ are delivered in a coded manner exploiting the cached contents as in the standard coded caching framework. The coded contents are transmitted over the BC using time-division among layers. We elaborate the placement and delivery phases of the CTDC scheme in Section \ref{CTDCScheme}.

\begin{prop}\label{PropAchievableRateCTDC}
For the system described in Section \ref{SystemModel} with average power $P$ and total cache capacity $M$, the rate tuple $\left( R_1, ..., R_K \right)$ is achievable by the CTDC scheme, if there exist $t_1, \ldots, t_K$, where $t_l \in [0:K-l]$, $\forall l \in [K]$, non-negative $R^{(1)}, \ldots, R^{(K)}$, and non-negative $\lambda^{(1)}, \ldots, \lambda^{(K)}$, such that $R_k = \sum\nolimits_{l=1}^{k} R^{(l)}$, $\sum\nolimits_{l=1}^{K} \lambda^{(l)} = 1$, $\forall k \in [K]$, and
\begin{subequations}
\label{AchievableRatePropCTDC}
\begin{align}\label{AchievableRatePropCTDCRate}
R^{(l)} &\le \lambda^{(l)} \frac{\sum\limits_{i=1}^{\binom{K-l+1}{t_l}} \prod\limits_{k \in \mathcal{K}_l \backslash \mathcal{S}_{\mathcal{K}_l,i}^{t_l}} C^P_{\sigma_k^2}}{\sum\limits_{i=1}^{\binom{K-l+1}{t_l+1}} \prod\limits_{k \in \mathcal{K}_l \backslash \mathcal{S}_{\mathcal{K}_l,i}^{t_l+1}} C^P_{\sigma_k^2}}, \quad \mbox{for $l \in [K]$},\\
M &= N \sum\limits_{l=1}^{K} t_l R^{(l)}.\label{AchievableRatePropCTDCCache}
\end{align}
\end{subequations}
\end{prop}

\begin{corollary}\label{AchievableRateRegionCorCTDC}
The following rate region for a total cache capacity $M$ and average power $P$ can be achieved by the CTDC scheme: 
\begin{align}\label{AchievableRateRegionTheoremCTDC}
\mathcal{C}_b (P,M) =& \bigcup\limits_{\lambda^{(1)}, \dots, \lambda^{(K)}: \sum\nolimits_{l=1}^{K} \lambda^{(l)} = 1} \left( \left\{ R_{1}, \dots,R_{K}  \right\}: \mbox{$\left( R_{1}, \dots,R_{K} \right)$ and $M$ satisfy \eqref{AchievableRatePropCTDC}} \right).
\end{align}
\end{corollary}

\begin{remark}\label{ConvexityCTDCRemark}
Let $(\hat{R}_{1}, \dots, \hat{R}_{K}) \in \mathcal{C}_b (P,M)$ and $(\tilde{R}_{1}, \dots, \tilde{R}_{K}) \in \mathcal{C}_b (P,M)$. Then, for any $\lambda \in [0,1]$, $(\lambda \hat{R}_{1} + \bar{\lambda} \tilde{R}_{1}, \dots, \lambda \hat{R}_{K} + \bar{\lambda} \tilde{R}_{K}) \in \mathcal{C}_b (P,M)$. This can be shown by joint time and memory-sharing. The whole library is divided into two parts according to $\lambda$, and the delivery of the two parts are carried out over two orthogonal time intervals of length $\lambda n$ and $\bar{\lambda} n$ using the codes for the two achievable tuples. Thus, for a fixed total cache capacity $M$, the rate pairs in the convex-hull of $\mathcal{C}_b (P,M)$ are achievable. 
\end{remark}

According to the convexity of the rate region $\mathcal{C}_b (P,M)$, a rate vector ${\textbf{R}}^* \buildrel \Delta \over = \left( R_{1}^*, \dots,R_{K}^*  \right)$ is on the boundary surface of $\mathcal{C}_b (P,M)$, if there exist non-negative coefficients $w_1, \dots, w_K$, $\sum\nolimits_{i=1}^{K} w_i = 1$, for which ${\textbf{R}}^*$ is a solution to the following optimization problem:
\begin{align}\label{BasicOptimization}
& \mathop {\max }\limits_{\lambda^{(1)}, \dots, \lambda^{(K)},R_1, \dots, R_K} \sum\limits_{i=1}^{K} w_i R_i, \nonumber\\
& \mbox{ subject to $\left\{ R_1, ..., R_K \right\} \in \mathcal{C}_b (P,M)$}. 
\end{align}
In the other words, for given weights $w_1, \dots, w_K$, and total cache capacity $M$, ${\textbf{R}}^*$ solves the problem in \eqref{BasicOptimization}, if $R^{(1)}, \dots, R^{(K)}$ is a solution of the following problem:
\begin{align}\label{DetailedOptimizationCTDC}
&\mathop {\max }\limits_{\lambda^{(1)}, \dots, \lambda^{(K)}, R^{(1)}, \dots, R^{(K)}} \sum\limits_{i=1}^{K} w_i \sum\limits_{l=1}^{i} R^{(l)}, \nonumber \\
& \mbox{subject to \eqref{AchievableRatePropCTDCRate} and \eqref{AchievableRatePropCTDCCache}},\nonumber\\
& \sum\limits_{l=1}^{K} \lambda^{(l)} = 1,
\end{align}
and 
\begin{align}\label{OptimizationOptimalSolutionCTDC}
R_{k}^* = \sum\limits_{l=1}^{k} R^{(l)}, \quad \mbox{for $k=1, \dots, K$}.  
\end{align}

\begin{remark}\label{CTDCRemConvex}
For given weights $w_1, \dots, w_K$, it is easy to verify that the problem in \eqref{DetailedOptimizationCTDC} is a linear optimization problem; thus it is a convex optimization problem. 
\end{remark}

\subsection{Joint Cache and Superposition Coding (CSC) and Joint Cache and Dirty Paper Coding (CDPC)}\label{CSCCDPCMainResults}
Here we present the achievable rate regions for the CSC and CDPC schemes. We introduce $r_1$ and $r_2$ to distinguish between the two, where we set $r_1=0$ and $r_2=1$ for CSC, while $r_1=1$ and $r_2=0$ for CDPC. We briefly highlight here that, with the CSC scheme, the coded packets of different layers are delivered over the Gaussian BC through superposition coding, while the CDPC scheme uses dirty paper coding to deliver the coded packets of different layers. The CSC scheme along with an example highlighting the main techniques and the CDPC scheme are elaborated in Section \ref{CSCCDPCSchemeProof}.

\begin{theorem}\label{AchievablePairsCSCSchemeTheorem}
For the system described in Section \ref{SystemModel} with average power $P$ and total cache capacity $M$, rate tuple $\left( R_1, ..., R_K \right)$ is achievable, if there exist $t \in [K-1]$, and non-negative $R^{(1)}, \ldots, R^{(K)}$, such that $R_k = \sum\nolimits_{l=1}^{k} R^{(l)}$, for $k \in [K]$, and
\begin{subequations}
\label{AchievableRateTheoremCSCCDPC}
\begin{align}\label{AchievableRateTheoremCSCCDPCLaye1}
    R^{(l)} = 
    \begin{cases} 
\sum\limits_{i=1}^{\binom{K}{t}} R^{(1)}_{\mathcal{S}_{\mathcal{K}_1,i}^{t}}, &\mbox{if $l=1$},\\
\sum\limits_{i=1}^{\binom{K-l+1}{t-1}} R^{(l)}_{\mathcal{S}_{\mathcal{K}_{l},i}^{t-1}}, &\mbox{if $l=2, ..., K-t+1$},\\
0, &\mbox{otherwise}.
\end{cases}
\end{align}
and, for $i \in \left[ 1+g_{l-1}:g_l \right]$ and $l \in [K-t]$,
\begin{align}\label{AchievableRateTheoremCSCCDPCleLayer1}
R_{\mathcal{S}_{\mathcal{K}_1,i}^{t+1} \backslash \{ k_1 \}}^{(1)} &\le \lambda_i C_{\bar{\alpha}_i P r_2+ \sigma_{k_1}^2}^{\alpha_i P}, \; \forall k_1 \in \mathcal{S}_{\mathcal{K}_1,i}^{t+1},\\
R_{\mathcal{S}_{\mathcal{K}_{l+1},i-g_{l-1}}^{t} \backslash \{ k_2 \}}^{(l+1)} &\le \lambda_i C_{{\alpha}_i P r_1+ \sigma_{k_2}^2}^{\bar{\alpha}_i P}, \; \forall k_2 \in \mathcal{S}_{\mathcal{K}_{l+1},i-g_{l-1}}^{t},\label{AchievableRateTheoremCSCCDPCleHigherLayers}
\end{align}
and 
\begin{align}\label{AchievableRateTheoremCSCCDPCCache}
M = N \left( t {R^{\left( 1 \right)}} + (t-1) \sum\limits_{l=2}^{K-t+1} {R^{\left( l \right)}} \right),
\end{align}
for some  
\begin{align}\label{AchievableRateTheoremCSCCDPCalphaiAndni}
0 \le \alpha_i &\le 1, \quad \mbox{for $i= 1, ..., \binom{K}{t+1}$},\\
0 \le \lambda_i &\le 1, \quad \mbox{for $i= 1, ..., \binom{K}{t+1}$},\label{AchievableRateTheoremCSCCDPCalphaiAndni2}\\
\label{AchievableRateTheoremCSCCDPCalphaiAndnisumni}
\sum\limits_{i=1}^{\binom{K}{t+1}} \lambda_i &= 1.
\end{align}
\end{subequations}
\end{theorem}

\begin{corollary}\label{AchievableRateRegionCor}
The following rate region for a total cache capacity $M$ and average power constraint $P$ can be achieved: 
\begin{align}\label{AchievableRateRegionTheoremCSCCDPC}
\mathcal{C}_c (P,M) = & \bigcup\limits_{\boldsymbol{\alpha},\boldsymbol{\lambda}:\sum\limits_{i=1}^{\binom{K}{t+1}} \lambda_i = 1} \left( \left\{ R_{1}, \dots,R_{K}  \right\}: \mbox{$\left( R_{1}, \dots,R_{K} \right)$ and $M$ satisfy \eqref{AchievableRateTheoremCSCCDPC}} \right),
\end{align}
where $\boldsymbol{\alpha} \buildrel \Delta \over = \alpha_1, \dots, \alpha_{\binom{K}{t+1}}$, and $\boldsymbol{\lambda} \buildrel \Delta \over = \lambda_1, \dots, \lambda_{\binom{K}{t+1}}$. 
\end{corollary}


For a fixed total cache capacity $M$, the convexity of region $\mathcal{C}_c (P,M)$ is followed through the same argument as Remark \ref{ConvexityCTDCRemark}, for both the CSC and CDPC schemes. As a result, for a given total cache capacity $M$, and for given non-negative coefficients $w_1, \dots, w_K$, such that $\sum\nolimits_{i=1}^{K} w_i = 1$, a rate vector ${\textbf{R}}^*$ is on the boundary surface of the achievable rate region $\mathcal{C}_c (P,M)$, if $R^{(1)}, \dots, R^{(K)}$ is a solution of the following problem:
\begin{subequations}
\label{VectorRsuperlDef}
\begin{align}\label{DetailedOptimization}
&\mathop {\max }\limits_{\boldsymbol{\alpha},\boldsymbol{\lambda},{\textbf{R}}^{(1)}, \dots, {\textbf{R}}^{(K-t+1)}} \sum\nolimits_{i=1}^{K} w_i \sum\nolimits_{l=1}^{i} R^{(l)}, \nonumber\\
& \mbox{subject to $R^{(1)}, \dots, R^{(K-t+1)}$ satisfy \eqref{AchievableRateTheoremCSCCDPCLaye1}},\nonumber \\
& \qquad \qquad \; \mbox{${\textbf{R}}^{(1)}$ satisfy \eqref{AchievableRateTheoremCSCCDPCleLayer1}},\nonumber \\
& \qquad \qquad \; \mbox{${\textbf{R}}^{(2)},\dots,{\textbf{R}}^{(K-t+1)}$ satisfy \eqref{AchievableRateTheoremCSCCDPCleHigherLayers}},\nonumber \\
& \qquad \qquad \; \mbox{$M$ satisfies \eqref{AchievableRateTheoremCSCCDPCCache}},\nonumber\\
& \qquad \qquad \; \mbox{$\boldsymbol{\alpha}$ and $\boldsymbol{\lambda}$ satisfy \eqref{AchievableRateTheoremCSCCDPCalphaiAndni}-\eqref{AchievableRateTheoremCSCCDPCalphaiAndnisumni}},
\end{align}
where 
\begin{align}\label{VectorRsuperlDef1}
{\textbf{R}}^{(1)}& \buildrel \Delta \over = R^{(1)}_{\mathcal{S}_{\mathcal{K}_1,1}^{t}}, \dots, R^{(1)}_{\mathcal{S}_{\mathcal{K}_1,\binom{K}{t}}^{t}},\\
{\textbf{R}}^{(l)}& \buildrel \Delta \over = R^{(l)}_{\mathcal{S}_{\mathcal{K}_l,1}^{t}}, \dots, R^{(l)}_{\mathcal{S}_{\mathcal{K}_l,\binom{K-l+1}{t-1}}^{t}}, \; \mbox{for $l \in [2:K-t+1]$},  
\end{align}
\end{subequations}
and 
\begin{align}\label{OptimizationOptimalSolution}
R_{k}^* = \sum\limits_{l=1}^{k} R^{(l)}, \quad \mbox{for $k=1, \dots, K$}.  
\end{align}

\begin{remark}\label{TimeSharingRemark}
Let $\tilde{\emph{\textbf{R}}} \buildrel \Delta \over = ( \tilde{R}_{1}, \dots,\tilde{R}_{K}  )$ and $\hat{\emph{\textbf{R}}} \buildrel \Delta \over = ( \hat{R}_{1}, \dots,\hat{R}_{K}  )$ be two achievable rate tuples for total cache capacities $\tilde{M}$ and $\hat{M}$, respectively. Then, $\beta \tilde{\textbf{R}} + \bar{\beta}\hat{\textbf{R}}$ can be achieved through joint time and memory-sharing for a total cache capacity $\beta \tilde{M} + \bar{\beta}\hat{M}$, for some $\beta \in [0,1]$. For $M=0$, the system under consideration is equivalent to the Gaussian BC without user caches, where user $k$ requests a file of rate $\sum\nolimits_{l = 1}^k R^{(l)}$, $k \in [K]$, and rate tuple $\emph{\textbf{R}}_z \buildrel \Delta \over = ( {R}_{z_1}, ..., {R}_{z_K} )$ is achievable by superposition coding, where
\begin{equation}\label{AchievableRatesZeroCacheCapacity}
    {R}_{z_k} = C_{\sum\limits_{i=k+1}^{K} {\gamma_i P} + \sigma_k^2}^{\gamma_kP}, \quad \mbox{for $k=1, ..., K$},
\end{equation}
for some non-negative coefficients $\gamma_1, \dots, \gamma_K$, such that $\sum\nolimits_{i=1}^{K} \gamma_i =1$. Hence, rate tuples $\beta \emph{\textbf{R}}_z + \bar{\beta}\tilde{\emph{\textbf{R}}}$ and $\beta \emph{\textbf{R}}_z + \bar{\beta}\hat{\emph{\textbf{R}}}$ are also achievable for total cache capacities $\bar{\beta} \tilde{M}$ and $\bar{\beta} \hat{M}$, respectively, through time sharing. 
\end{remark}

\section{Proof of Proposition \ref{PropAchievableRateCTDC}}\label{CTDCScheme}
With the DTM scheme, the layers with $\textbf{W}^{(l)}$, for $l \in [K]$, are cached and delivered via a distinct time slot (TS). We elaborate the placement and delivery phases of the CTDC scheme in the following. 

\underline{\textbf{Placement phase:}} The layers with $\textbf{W}^{(l)}$ are cached partially by the users in $\mathcal{K}_l$ constrained by their cache capacities, for $l \in [K]$. For caching factors $t_1, \ldots, t_K$, where $t_l \in [0:K-l]$, layer $W^{(l)}_j$ is divided into $\binom{K-l+1}{t_l}$ disjoint subfiles $W^{(l)}_{j,\mathcal{S}^{t_l}_{\mathcal{K}_l,1}}$, $\dots$, $W^{(l)}_{j,\mathcal{S}^{t_l}_{\mathcal{K}_l,\binom{K-l+1}{t_l}}}$, where subfile $W^{(l)}_{j,\mathcal{S}^{t_l}_{\mathcal{K}_l,i}}$ is of rate $R^{(l)}_{\mathcal{S}^{t_l}_{\mathcal{K}_l,i}}$, for $i \in \left[ \binom{K-l+1}{t_l} \right]$, $l \in [K]$, $j \in [N]$.\footnote{We assume throughout the paper that, for any real number $A \ge 0$, $2^{nA}$ is an integer for $n$ large enough.} We note that
\begin{align}\label{SumRateSubfilesCTDC}
R^{(l)} = \sum\limits_{i=1}^{\binom{K-l+1}{t_l}} R^{(l)}_{\mathcal{S}^{t_l}_{\mathcal{K}_l,i}}, \quad \mbox{for $l \in [K]$}.
\end{align}
User $k$'s cache content, for $k \in [K]$, is given by 
\begin{equation}\label{UserkCacheContentCTDC}
    U_k = \bigcup\limits_{j = 1}^{N} \bigcup\limits_{l = 1}^{k} \bigcup\limits_{i \in \left[ \binom{K-l+1}{t_l} \right] : k \in \mathcal{S}^{t_l}_{\mathcal{K}_l,i}} {W_{j,\mathcal{S}_{\mathcal{K}_l,i}^{t_l}}^{(l)}}, 
\end{equation}
which leads to a total cache capacity of 
\begin{equation}\label{TotalCacheCapacityCTDC}
    M = \sum\limits_{k=1}^{K} M_k = N \sum\limits_{l=1}^{K} t_l R^{(l)}. 
\end{equation}

\underline{\textbf{Delivery phase:}} Given a demand vector $\textbf{d} = \left( d_1, ..., d_K \right)$, the server aims to deliver the coded packet
\begin{equation}\label{CodedPacketDefCTDC}
W^{(l)}_{\mathcal{S}_{\mathcal{K}_l,i}^{t_l+1}} \buildrel \Delta \over = {\overline{\bigoplus}} _{k \in \mathcal{S}_{\mathcal{K}_l,i}^{t_l+1}} W^{(l)}_{d_k,\mathcal{S}_{\mathcal{K}_l,i}^{t_l+1} \backslash \{k\}}
\end{equation}
of rate
\begin{equation}\label{CodedPacketRateDefCTDC}
R^{(l)}_{{\rm{XOR}},\mathcal{S}_{\mathcal{K}_l,i}^{t_l+1}} \buildrel \Delta \over = \mathop {\max }\limits_{k \in \mathcal{S}_{\mathcal{K}_l,i}^{t_l+1}} \left\{ R^{(l)}_{\mathcal{S}_{\mathcal{K}_l,i}^{t_l+1} \backslash \{ k\}} \right\} 
\end{equation}
to the users in $\mathcal{S}_{\mathcal{K}_l,i}^{t_l+1}$, for $i \in \left[ \binom{K-l+1}{t_l+1} \right]$ and $l \in [K]$. Each user $k \in \mathcal{K}_l$ can obtain all missing bits of its request $W^{(l)}_{d_k}$ after receiving 
\begin{equation}\label{ObtainingAllMissingBitslayerlCTDC}
\bigcup\limits_{i \in \left[ \binom{K-l+1}{t_l+1} \right]:k \in \mathcal{S}_{\mathcal{K}_l,i}^{t_l+1}} W^{(l)}_{\mathcal{S}_{\mathcal{K}_l,i}^{t_l+1}}      
\end{equation}
along with its cache content, for $l \in [K]$. We allocate a distinct $\lambda^{(l)}n$ channel uses to deliver the coded packets $W^{(l)}_{\mathcal{S}_{\mathcal{K}_l,1}^{t_l+1}}$,$\dots$,$W^{(l)}_{\mathcal{S}_{\mathcal{K}_l,\binom{K-l+1}{t_l+1}}^{t_l+1}}$ to the intended users in $\mathcal{K}_l$, for some $\lambda^{(l)} \in [0,1]$, $l \in [K]$, where each coded packet among them is delivered via a different TS, and we have $\sum\nolimits_{l=1}^{K} \lambda^{(l)} = 1$. The coded packet $W^{(l)}_{\mathcal{S}_{\mathcal{K}_l,i}^{t_l+1}}$ of the files in the $l$-th layer is delivered to the users $\mathcal{S}_{\mathcal{K}_l,i}^{t_l+1}$ through a distinct time interval of length $\lambda^{(l)}_i n$, for $i \in \left[ \binom{K-l+1}{t_l+1} \right]$ and $l \in [K]$, where $\sum\nolimits_{i=1}^{\binom{K-l+1}{t_l+1}} \lambda^{(l)}_i = \lambda^{(l)}$. In order to recover the coded packet $W^{(l)}_{\mathcal{S}_{\mathcal{K}_l,i}^{t_l+1}}$, user ${k_1} \in \mathcal{S}_{\mathcal{K}_l,i}^{t_l+1}$ first generates 
\begin{equation}\label{CTDCRecoverUsersFirstLayer}
    {\overline{\bigoplus}} _{k \in \mathcal{S}_{\mathcal{K}_l,i}^{t_l+1}\backslash \{ {k_1}\}} W^{(l)}_{d_k,\mathcal{S}_{\mathcal{K}_l,i}^{t_l+1} \backslash \{k\}}
\end{equation}
from its cache; it then only needs to decode $W^{(l)}_{d_{{k_1}},\mathcal{S}_{\mathcal{K}_l,i}^{t_l+1} \backslash \{{k_1}\}}$ of rate $R^{(l)}_{\mathcal{S}_{\mathcal{K}_l,i}^{t_l+1} \backslash \{k_1\}}$, which the decoding is successful for $n$ large enough, if
\begin{align}\label{CTDCRecoverUsersFirstLayerDecoding}
    R^{(l)}_{\mathcal{S}_{\mathcal{K}_l,i}^{t_l+1} \backslash \{{k_1}\}} \le \lambda^{(l)}_i C^{P}_{\sigma_{{k_1}}^2}, \quad &\mbox{for $i \in \left[ \binom{K-l+1}{t_l+1} \right]$ and $l \in [K]$}. 
\end{align}
By choosing 
\begin{align}\label{CTDCRecoverUsersnilDecoding}
    \lambda_i^{(l)} = \frac{\prod\limits_{k \in \mathcal{K}_l \backslash \mathcal{S}_{\mathcal{K}_l,i}^{t_l+1}} C^P_{\sigma_k^2}}{\sum\limits_{i=1}^{\binom{K-l+1}{t_l+1}} \prod\limits_{k \in \mathcal{K}_l \backslash \mathcal{S}_{\mathcal{K}_l,i}^{t_l+1}} C^P_{\sigma_k^2}} \lambda^{(l)}, \quad &\mbox{for $i \in \left[ \binom{K-l+1}{t_l+1} \right]$ and $l \in [K]$}, 
\end{align}
which satisfies $\sum\nolimits_{i=1}^{\binom{K-l+1}{t_l+1}} \lambda^{(l)}_i = \lambda^{(l)}$ and leads to
\begin{align}\label{CTDCRecoverUsersRilDecoding}
    R^{(l)}_{\mathcal{S}_{\mathcal{K}_l,i}^{t_l}} \le \lambda^{(l)} \frac{\prod\limits_{k \in \mathcal{K}_l \backslash \mathcal{S}_{\mathcal{K}_l,i}^{t_l}} C^P_{\sigma_k^2}}{\sum\limits_{i=1}^{\binom{K-l+1}{t_l+1}} \prod\limits_{k \in \mathcal{K}_l \backslash \mathcal{S}_{\mathcal{K}_l,i}^{t_l+1}} C^P_{\sigma_k^2}}, 
\end{align}
it can be checked that all the conditions in \eqref{CTDCRecoverUsersFirstLayerDecoding} are satisfied. Therefore, the coded packets $W^{(l)}_{\mathcal{S}_{\mathcal{K}_l,1}^{t_l+1}}$,$\dots$,$W^{(l)}_{\mathcal{S}_{\mathcal{K}_l,\binom{K-l+1}{t_l+1}}^{t_l+1}}$, each delivered with an average power $P$ via a distinct TS, can be decoded by their intended users successfully, if, for $n$ large enough,
\begin{align}\label{layerlDecodeSuccessfulCTDC}
R^{(l)} \le \lambda^{(l)} \frac{\sum\limits_{i=1}^{\binom{K-l+1}{t_l}} \prod\limits_{k \in \mathcal{K}_l \backslash \mathcal{S}_{\mathcal{K}_l,i}^{t_l}} C^P_{\sigma_k^2}}{\sum\limits_{i=1}^{\binom{K-l+1}{t_l+1}} \prod\limits_{k \in \mathcal{K}_l \backslash \mathcal{S}_{\mathcal{K}_l,i}^{t_l+1}} C^P_{\sigma_k^2}}, \quad \mbox{for $l \in [K]$},     
\end{align}
which together with the total cache capacity given in \eqref{TotalCacheCapacityCTDC} complete the proof of Proposition \ref{PropAchievableRateCTDC}.

We remark here that the CTDC scheme applies the scheme in \cite{ShirinWiggerYenerCacheAssingment}, proposed when the messages are of the same rate, to the scenario of multiple-layer messages through joint time and memory-sharing.

\section{Proof of Theorem \ref{AchievablePairsCSCSchemeTheorem}}\label{CSCCDPCSchemeProof}
Here we present the CSC and CDPC schemes, which achieve the rate tuple presented in Theorem \ref{AchievablePairsCSCSchemeTheorem} for $r_1=0,r_2=1$, and $r_1=1,r_2=0$, respectively, and $t \in [K-1]$.       

\underline{\textbf{Placement phase:}} As described in Section \ref{SystemModel}, user $k$ receives layers 1 to $k$ of its request, i.e., $W_{d_k}^{(1)}, \dots, W_{d_k}^{(k)}$, for $k \in [K]$. Thus, the $l$-th layer of the files, i.e., $\textbf{W}^{(l)}$, are cached partially by the users in $\mathcal{K}_l$ constrained by their cache capacities, for $l \in [K]$. For $t \in [K-1]$, we set
\begin{align}\label{DeftlayersProposedScheme}
    t_l = \begin{cases} 
t, &\mbox{if $l=1$},\\
t-1, &\mbox{if $2 \le l \le K-t+1$},\\
0 &\mbox{otherwise},
\end{cases}
\end{align}
and 
\begin{align}\label{CSCCDPCRateLyaerlZero}
R^{(l)} = 0, \quad \mbox{for $l \in \left[ K-t+2:K \right]$}.  
\end{align}
The $l$-th layer of each file, i.e., each layer with $\textbf{W}^{(l)}$, which are targeted for users in $\mathcal{K}_l$, is split into $\binom{K-l+1}{t_l}$ disjoint subfiles, for $l \in [K-t+1]$, represented by      
\begin{equation}\label{HighLevellayerSubfiles}
    W_{j}^{(l)} = \bigcup\limits_{i = 1}^{\binom{K-l+1}{t_l}} {W_{j,\mathcal{S}_{\mathcal{K}_l,i}^{t_l}}^{(l)}}, \quad \mbox{for $j \in [N]$},
\end{equation}
where subfile ${W_{j,\mathcal{S}_{\mathcal{K}_l,i}^{t_l}}^{(l)}}$ is of rate $R_{\mathcal{S}_{\mathcal{K}_l,i}^{t_l}}^{(l)}$, for $i \in \left[ \binom{K-l+1}{t_l} \right]$. We note that $\sum\nolimits_{i=1}^{\binom{K-l+1}{t_l}} {R_{\mathcal{S}_{\mathcal{K}_l,i}^{t_l}}^{(l)}} = R^{(l)}$, for $l \in [K-t+1]$. User $k$'s cache content, $k \in [K]$, is given by
\begin{equation}\label{UserkCacheContent}
    U_k = \bigcup\limits_{j = 1}^{N} \bigcup\limits_{l = 1}^{k} \bigcup\limits_{i \in \left[ \binom{K-l+1}{t_l} \right] : k \in \mathcal{S}^{t_l}_{\mathcal{K}_l,i}} {W_{j,\mathcal{S}_{\mathcal{K}_l,i}^{t_l}}^{(l)}}, 
\end{equation}
leading to the cache capacity
\begin{equation}\label{UserkCacheCapacity}
    M_k = N \sum\limits_{l = 1}^{k} \sum\limits_{i \in \left[ \binom{K-l+1}{t_l} \right] : k \in \mathcal{S}^{t_l}_{\mathcal{K}_l,i}} {R_{\mathcal{S}_{\mathcal{K}_l,i}^{t_l}}^{(l)}}. 
\end{equation}
We can obtain the total cache capacity in the system as
\begin{equation}\label{TotalCacheCapacity}
    M = \sum\limits_{k=1}^{K} M_k = N \sum\limits_{l=1}^{K-t+1} t_lR^{(l)} = N \left( t {R^{\left( 1 \right)}} + (t-1) \sum\limits_{l=2}^{K-t+1} {R^{\left( l \right)}} \right), 
\end{equation}
which is equal to the one in \eqref{AchievableRateTheoremCSCCDPCCache}. We note that the rate of the $l$-th layer of each file, i.e., $R^{(l)}$, for $l \in [K]$, corresponds to \eqref{AchievableRateTheoremCSCCDPCLaye1}.

\underline{\textbf{Delivery phase:}} Given a demand vector $\textbf{d} = \left( d_1, ..., d_K \right)$, the server delivers the coded packet 
\begin{equation}\label{CodedPacketDef}
W^{(l)}_{\mathcal{S}_{\mathcal{K}_l,i}^{t_l+1}} = {\overline{\bigoplus}} _{k \in \mathcal{S}_{\mathcal{K}_l,i}^{t_l+1}} W^{(l)}_{d_k,\mathcal{S}_{\mathcal{K}_l,i}^{t_l+1} \backslash \{k\}}, \; \mbox{for $i \in \left[ \binom{K-l+1}{t_l+1}\right]$}, 
\end{equation}
of rate 
\begin{equation}\label{CodedPacketRateDef}
R^{(l)}_{{\rm{XOR}},\mathcal{S}_{\mathcal{K}_l,i}^{t_l+1}} = \mathop {\max }\limits_{k \in \mathcal{S}_{\mathcal{K}_l,i}^{t_l+1}} \left\{ R^{(l)}_{\mathcal{S}_{\mathcal{K}_l,i}^{t_l+1} \backslash \{ k\}} \right\} 
\end{equation}
to the users in $\mathcal{S}_{\mathcal{K}_l,i}^{t_l}$, for $l \in [K-t+1]$. Thus, after receiving $W^{(l)}_{\mathcal{S}_{\mathcal{K}_l,i}^{t_l+1}}$, each user $k \in {\mathcal{S}_{\mathcal{K}_l,i}^{t_l+1}}$ can recover the missing subfile $W^{(l)}_{d_k,\mathcal{S}_{\mathcal{K}_l,i}^{t_l+1} \backslash \{k\}}$ of the $l$-th layer of its request, for $i \in \left[ \binom{K-l+1}{t_l+1}\right]$ and $l \in [K-t+1]$. We note that, user $k$, for $k \in [K]$, only exists in the sets $\mathcal{K}_1, \dots, \mathcal{K}_k$, and also the rate of each layer with $\textbf{W}^{(K-t+2)}, \dots, \textbf{W}^{(K)}$ is set to zero. Thus, user $k$, for $k \in [K-t+1]$, can recover all missing subfiles of layers $W_{d_k}^{(1)}, \dots, W_{d_k}^{(k)}$ after receiving all the coded packets $W^{(l)}_{\mathcal{S}_{\mathcal{K}_l,i}^{t_l+1}}$, $\forall i \in \left[ \binom{K-l+1}{t_l+1}\right]$ and $\forall l \in [k]$, such that $k \in \mathcal{S}_{\mathcal{K}_l,i}^{t_l+1}$. On the other hand, user $k$, for $k \in [K-t+2:K]$, can recover the missing bits of all the layers $W_{d_k}^{(1)}, \dots, W_{d_k}^{(K-t+1)}$ after receiving $W^{(l)}_{\mathcal{S}_{\mathcal{K}_l,i}^{t_l+1}}$, $\forall i \in \left[ \binom{K-l+1}{t_l+1}\right]$ and $\forall l \in [K-t+1]$, such that $k \in \mathcal{S}_{\mathcal{K}_l,i}^{t_l+1}$. We remind here that $t_l$ is given in \eqref{DeftlayersProposedScheme}.

The main technique to deliver the coded packets is to send the packet targeted to the users in ${\mathcal{S}_{\mathcal{K}_{l+1}, i}^{t}}$ along with the packet targeted to the users in ${\mathcal{S}_{\mathcal{K}_{1}, i+ g_{l-1}}^{t+1}}$ through different channel coding techniques, where, from \eqref{LabelSubsetsInsideeachGroupK1Kk}, ${\mathcal{S}_{\mathcal{K}_{l+1}, i}^{t}} = {\mathcal{S}_{\mathcal{K}_{1}, i+ g_{l-1}}^{t+1}} \backslash \{ l \}$, for $i \in \left[ \binom{K-l}{t} \right]$ and $l \in [K-t]$. For this purpose, the transmission is performed via $\binom{K}{t+1}$ orthogonal TSs, where the $i$-th TS is of length $\lambda_i n$ channel uses, for $i \in \left[ \binom{K}{t+1} \right]$, so that $\sum\nolimits_{i=1}^{\binom{K}{t+1}} \lambda_i = 1$.

\subsection{Joint Cache and Superposition Coding (CSC) Scheme}\label{CSCScheme}
With TS $i$, for $i \in \left[1+g_{l-1}:g_{l}\right]$ and $l \in [K-t]$, we generate two subcodebooks
\begin{subequations}\label{CSCCodebooksDef}
\begin{align}\label{CSCCodebooksDef1}
\mathcal{C}_1& \buildrel \Delta \over= \left\{ x_1^{\lambda_{i}n} \left( w_{1} \right): w_{1} \in \left[ 2^{n R^{(1)}_{{\rm{XOR}},\mathcal{S}_{\mathcal{K}_1,i}^{t+1}}} \right] \right\},\\
\mathcal{C}_2& \buildrel \Delta \over= \left\{ x_2^{\lambda_{i}n} \left( w_{2} \right): w_{2} \in \left[ 2^{n R^{(l+1)}_{{\rm{XOR}},\mathcal{S}_{\mathcal{K}_{l+1},i-g_{l-1}}^{t}}} \right] \right\},\label{CSCCodebooksDef2}
\end{align}
\end{subequations}
where all the entries in $\mathcal{C}_1$ and $\mathcal{C}_2$ are drawn i.i.d. according to $\mathcal{N} \left( 0,\alpha_i P \right)$ and $\mathcal{N} \left( 0,\bar{\alpha}_i P \right)$, respectively, for some $0 \le \alpha_i \le 1$. The server then transmits the codeword
\begin{align}\label{CSCServerTransmittedSignal}
    x_1^{\lambda_{i}n} \left( W_{\mathcal{S}_{\mathcal{K}_1,i}^{t+1}}^{(1)} \right) + x_2^{\lambda_{i}n} \left( W_{\mathcal{S}_{\mathcal{K}_{l+1},i-g_{l-1}}^{t}}^{(l+1)} \right),
\end{align}
sent through linear superposition of the codewords from subcodebooks $\mathcal{C}_1$ and $\mathcal{C}_2$, over the Gaussian BC with TS $i$, for $i \in \left[1+g_{l-1}:g_{l}\right]$, $l \in [K-t]$. We note that $g_{K-t} = \sum\nolimits_{j=1}^{K-t} \binom{K-j}{t} = \binom{K}{t+1}$. We also note that, if all the coded packets $W_{\mathcal{S}_{\mathcal{K}_{l+1},i-g_{l-1}}^{t}}^{(l+1)}$ are received by their targeted users successfully via all TSs $i$, $\forall i \in \left[1+g_{l-1}:g_{l}\right]$, then the users in $\mathcal{K}_{l+1}$ can obtain the missing subfiles of the $(l+1)$-th layer of their demands, for $l \in [K-t]$. On the other hand, the users in $\mathcal{K}_{1}$ need to receive all coded packets $W_{\mathcal{S}_{\mathcal{K}_1,i}^{t+1}}^{(1)}$ targeted for them via all $\binom{K}{t+1}$ TSs to obtain the first layer of their requests. The users in $\mathcal{S}_{\mathcal{K}_1,i}^{t+1}$ first decode the message with $x_1^{\lambda_in}$, while considering $x_2^{\lambda_in}$ as noise. To decode the message with $x_1^{\lambda_in}$, each user $k_1 \in \mathcal{S}_{\mathcal{K}_1,i}^{t+1}$ first recovers 
\begin{equation}\label{CSCRecoverUsersFirstLayer}
    {\overline{\bigoplus}} _{k \in \mathcal{S}_{\mathcal{K}_1,i}^{t+1}\backslash \{ k_1\}} W^{(1)}_{d_k,\mathcal{S}_{\mathcal{K}_1,i}^{t+1} \backslash \{k\}}
\end{equation}
from its cache, and it only needs to decode $W^{(1)}_{d_{k_1},\mathcal{S}_{\mathcal{K}_1,i}^{t+1} \backslash \{k_1\}}$ of rate $R^{(1)}_{\mathcal{S}_{\mathcal{K}_1,i}^{t+1} \backslash \{k_1\}}$, which, using an optimal decoding, the decoding is successful for $n$ large enough, if
\begin{align}\label{CSCRecoverUsersFirstLayerDecoding}
    R^{(1)}_{\mathcal{S}_{\mathcal{K}_1,i}^{t+1} \backslash \{k_1\}} \le \lambda_i C^{\alpha_iP}_{\bar{\alpha}_i P + \sigma_{k_1}^2}, \quad &\mbox{for $i \in \left[1+g_{l-1}:g_{l}\right]$ and $l \in [K-t]$}. 
\end{align}
We note, from \eqref{LabelSubsetsInsideeachGroupK1Kk}, that 
\begin{align}\label{LabelSubsetsInsideeachGroupK1KkCSC}
    {\mathcal{S}_{\mathcal{K}_{l+1}, i-g_{l-1}}^{t}} = {\mathcal{S}_{\mathcal{K}_{1}, i}^{t+1}} \backslash \{ l \} , \quad \mbox{for $i \in \left[1+g_{l-1}:g_{l}\right]$}; 
\end{align}
thus, each user in ${\mathcal{S}_{\mathcal{K}_{l+1}, i-g_{l-1}}^{t}}$, for which the message with codeword $x_2^{\lambda_i n}$ is targeted to, can decode $x_1^{\lambda_i n}$ having \eqref{CSCRecoverUsersFirstLayerDecoding} satisfied, for $l \in [K-t]$. Similarly, to decode the message with $x_2^{\lambda_in}$, each user $k_2 \in \mathcal{S}_{\mathcal{K}_{l+1},i-g_{l-1}}^{t}$ first recovers 
\begin{equation}\label{CSCRecoverUsersHigherLayer}
    {\overline{\bigoplus}} _{k \in \mathcal{S}_{\mathcal{K}_{l+1},i-g_{l-1}}^{t}\backslash \{ k_2\}} W^{(l+1)}_{d_k,\mathcal{S}_{\mathcal{K}_{l+1},i-g_{l-1}}^{t} \backslash \{k\}}
\end{equation}
from its cache, and it only needs to decode $W^{(l+1)}_{d_{k_2},\mathcal{S}_{\mathcal{K}_{l+1},i-g_{l-1}}^{t} \backslash \{k_2\}}$ of rate $R^{(l+1)}_{\mathcal{S}_{\mathcal{K}_{l+1},i-g_{l-1}}^{t} \backslash \{k_2\}}$, which, using an optimal decoding, the decoding is successful for $n$ large enough, if
\begin{align}\label{CSCRecoverUsersHigherLayerDecoding}
    R^{(l+1)}_{\mathcal{S}_{\mathcal{K}_{l+1},i-g_{l-1}}^{t} \backslash \{k_2\}} \le \lambda_i C^{\bar{\alpha}_iP}_{\sigma_{k_2}^2}, \quad &\mbox{for $i \in \left[1+g_{l-1}:g_{l}\right]$ and $l \in [K-t]$}. 
\end{align}

Observe that the conditions in \eqref{CSCRecoverUsersFirstLayerDecoding} and \eqref{CSCRecoverUsersHigherLayerDecoding} prove the achievability of the rate tuple outlined in Theorem \ref{AchievablePairsCSCSchemeTheorem} for the total cache capacity $M$ given in \eqref{TotalCacheCapacity}, which is the same as the one in \eqref{AchievableRateTheoremCSCCDPCCache}, when $r_1=0$ and $r_2=1$.   

In the following, we present an example of the CSC scheme for more clarification.

\begin{table}[!t]
\caption{Codewords sent via 4 TSs in the delivery phase.}
\centering
\begin{tabular}{ |c|c| }
\hline
TS number & Transmitted codeword\\
\hline
1 & $x_1^{\lambda_{1}n} \left( W_{\{1,2,3 \}}^{(1)} \right) + x_2^{\lambda_{1}n} \left( W_{\{ 2,3 \}}^{(2)} \right)$\\
\hline
2 & $x_1^{\lambda_{2}n} \left( W_{\{1,2,4 \}}^{(1)} \right) + x_2^{\lambda_{2}n} \left( W_{\{ 2,4 \}}^{(2)} \right)$\\
\hline
3 & $x_1^{\lambda_{3}n} \left( W_{\{1,3,4 \}}^{(1)} \right) + x_2^{\lambda_{3}n} \left( W_{\{ 3,4 \}}^{(2)} \right)$\\
\hline
4 & $x_1^{\lambda_{4}n} \left( W_{\{2,3,4 \}}^{(1)} \right) + x_2^{\lambda_{4}n} \left( W_{\{ 3,4 \}}^{(3)} \right)$\\
\hline
\end{tabular}
\label{TableDeliveryExample1}
\end{table}

\subsection{Example}\label{CSCSectionExample}
Consider a cache-aided network as described in Section \ref{SystemModel} with $K =4$ users in the system. Here we exemplify the achievability of the rate region stated in Theorem \ref{AchievablePairsCSCSchemeTheorem} for the CSC scheme for $t=2$. We set $R^{(4)}=0$, and split the messages in the $l$-th layer, for $l \in [3]$, as follows:
\begin{subequations}
\label{layersExmp}
\begin{align}\label{layersExmp1}
W_j^{\left( 1 \right)}  =& \left( {W_{j,\{ 1,2 \}}^{\left( 1 \right)}},{W_{j,\{ 1,3 \}}^{\left( 1 \right)}},{W_{j,\{ 1,4 \}}^{\left( 1 \right)}},{W_{j,\{ 2,3 \}}^{\left( 1 \right)}},{W_{j,\{ 2,4 \}}^{\left( 1 \right)}},{W_{j,\{ 3,4 \}}^{\left( 1 \right)}}\right),\\
W_j^{\left( 2 \right)} =& \left( {W_{j,\{ 2 \}}^{\left( 2 \right)}},{W_{j,\{ 3 \}}^{\left( 2 \right)}},{W_{j,\{ 4 \}}^{\left( 2 \right)}}\right), 
\label{layersExmp2}\\
W_j^{\left( 3 \right)} =& \left({W_{j,\{ 3 \}}^{\left( 3 \right)}}, {W_{j,\{ 4 \}}^{\left( 3 \right)}}\right),
\label{layersExmp3}
\end{align}
\end{subequations}
where subfile ${W_{j,\mathcal{S}_{\mathcal{K}_l,i}^{t_l}}^{(l)}}$ is of rate $R_{\mathcal{S}_{\mathcal{K}_l,i}^{t_l}}^{(l)}$, for $i \in \left[ \binom{5-l}{t_l} \right]$, $\forall j \in [N]$, and $t_1=2, t_2=t_3=1$, and $t_4=0$.

\begin{table}[!t]
\caption{Decoding the message with $x_1^{\lambda_{i} n}$ at TS $i$, for $i=1, ..., 4$}
\centering
\begin{tabular}{ |c|c| }
\hline
TS number & Sufficient conditions\\
\hline
\multirow{3.8}{*}{1} & \vtop{\hbox{\strut $ R^{(1)}_{\{2,3 \}} \le \lambda_1 C_{\bar{\alpha}_1P+\sigma^2_1}^{\alpha_1P}$} \hbox{\strut $R^{(1)}_{\{1,3 \}} \le \lambda_1 C_{\bar{\alpha}_1P+\sigma^2_2}^{\alpha_1P}$} \hbox{\strut $R^{(1)}_{\{1,2 \}} \le \lambda_1 C_{\bar{\alpha}_1P+\sigma^2_3}^{\alpha_1P}$}} \\
\hline
\multirow{3.8}{*}{2} & \vtop{\hbox{\strut $R^{(1)}_{\{2,4 \}} \le \lambda_2 C_{\bar{\alpha}_2P+\sigma^2_1}^{\alpha_2P}$} \hbox{\strut $R^{(1)}_{\{1,4 \}} \le \lambda_2 C_{\bar{\alpha}_2P+\sigma^2_2}^{\alpha_2P}$} \hbox{\strut $R^{(1)}_{\{1,2 \}} \le \lambda_2 C_{\bar{\alpha}_2P+\sigma^2_4}^{\alpha_2P}$}}\\
\hline
\multirow{3.8}{*}{3} & \vtop{\hbox{\strut $R^{(1)}_{\{3,4 \}} \le \lambda_3 C_{\bar{\alpha}_3P+\sigma^2_1}^{\alpha_3P}$} \hbox{\strut $R^{(1)}_{\{1,4 \}} \le \lambda_3 C_{\bar{\alpha}_3P+\sigma^2_3}^{\alpha_3P}$} \hbox{\strut $R^{(1)}_{\{1,3 \}} \le \lambda_3 C_{\bar{\alpha}_3P+\sigma^2_4}^{\alpha_3P}$}}\\
\hline
\multirow{3.8}{*}{4} & \vtop{\hbox{\strut $R^{(1)}_{\{3,4 \}} \le \lambda_4 C_{\bar{\alpha}_4P+\sigma^2_2}^{\alpha_4P}$} \hbox{\strut $R^{(1)}_{\{2,4 \}} \le \lambda_4 C_{\bar{\alpha}_4P+\sigma^2_3}^{\alpha_4P}$} \hbox{\strut $R^{(1)}_{\{2,3 \}} \le \lambda_4 C_{\bar{\alpha}_4P+\sigma^2_4}^{\alpha_4P}$}}\\
\hline
\end{tabular}
\label{TableDeliveryExample1Codeword1}
\end{table}

The cache content of each user is given by
\begin{subequations}
\label{CacheContentsExmp}
\begin{align}\label{CacheContentsExmp1}
{U_1} & = \bigcup\limits_{j \in \left[ N \right]} {\left( {W_{j,\left\{ 1,2 \right\}}^{\left( 1 \right)},W_{j,\left\{ {1,3} \right\}}^{\left( 1 \right)},W_{j,\left\{ {1,4} \right\}}^{\left( 1 \right)}} \right)},\\
{U_2} & = \bigcup\limits_{j \in \left[ N \right]} {\left( W_{j,\left\{ 1,2 \right\}}^{\left( 1 \right)},W_{j,\left\{ {2,3} \right\}}^{\left( 1 \right)},W_{j,\left\{ {2,4} \right\}}^{\left( 1 \right)}, W_{j,\left\{ {2} \right\}}^{\left( 2 \right)} \right)} ,\label{CacheContentsExmp2}\\
{U_3} & = \bigcup\limits_{j \in \left[ N \right]} {\left( W_{j,\left\{ 1,3 \right\}}^{\left( 1 \right)},W_{j,\left\{ {2,3} \right\}}^{\left( 1 \right)},W_{j,\left\{ {3,4} \right\}}^{\left( 1 \right)}, W_{j,\left\{ {3} \right\}}^{\left( 2 \right)}, W_{j,\left\{ {3} \right\}}^{\left( 3 \right)} \right)},
\label{CacheContentsExmp3}\\
{U_4} & = \bigcup\limits_{j \in \left[ N \right]} {\left( W_{j,\left\{ 1,4 \right\}}^{\left( 1 \right)},W_{j,\left\{ {2,4} \right\}}^{\left( 1 \right)},W_{j,\left\{ {3,4} \right\}}^{\left( 1 \right)}, W_{j,\left\{ {4} \right\}}^{\left( 2 \right)}, W_{j,\left\{ {4} \right\}}^{\left( 3 \right)} \right)},
\label{CacheContentsExmp4}
\end{align}
\end{subequations}
where the total cache capacity in the system is 
\begin{equation}\label{TotalCacheCapacityExmpl}
   M = N \left( 2R^{(1)}+R^{(2)}+R^{(3)} \right). 
\end{equation}

For a demand vector $\textbf{d} = \left( d_1, \dots, d_4 \right)$ in the delivery phase, we generate coded packet $W^{(l)}_{\mathcal{S}_{\mathcal{K}_l,i}^{t_l+1}}$, for $i \in \left[ \binom{5-l}{t_l+1}\right]$ and $l \in [3]$, as given in \eqref{CodedPacketDef}. The transmission is performed via 4 orthogonal TSs, where the $i$-th TS is of length $\lambda_in$ channel uses, for $i \in [4]$, such that $\sum\nolimits_{i=1}^{4} \lambda_i = 1$. After generating codebooks $\mathcal C_1$ and $\mathcal C_2$ in the $i$-th TS as defined in \eqref{CSCCodebooksDef}, for $i \in [4]$, the codeword outlined in Table \ref{TableDeliveryExample1} is sent over the channel via each TS. In TS $i$, the users, which the message with $x_1^{\lambda_{i}n}$ is targeted to, decode the message with $x_1^{\lambda_{i}n}$ while considering $x_2^{\lambda_{i}n}$ as noise, for $i \in [4]$. The sufficient conditions, for which the message with $x_1^{\lambda_{i}n}$ can be decoded successfully by each targeted user, for $n$ large enough, at TS $i$ are summarized in Table \ref{TableDeliveryExample1Codeword1}, for $i \in [4]$, thanks to the side information available at users' caches. We note that each user, for which the message with codeword $x_2^{\lambda_{i}n}$ is targeted to, can decode the message with $x_1^{\lambda_{i}n}$ having the conditions in Table \ref{TableDeliveryExample1Codeword1} satisfied, for $i \in [4]$. Bearing this in mind, the sufficient conditions such that the message with $x_2^{\lambda_{i}n}$ is decoded successfully by the intended users are outlined in Table \ref{TableDeliveryExample1Codeword23}, for $i \in [4]$. We note that the conditions in Table \ref{TableDeliveryExample1Codeword1} and Table \ref{TableDeliveryExample1Codeword23} guarantee the achievability of the rate tuple presented in Theorem \ref{AchievablePairsCSCSchemeTheorem} for the corresponding total cache capacity $M$ given in \eqref{TotalCacheCapacityExmpl}, which is equivalent to the one in \eqref{AchievableRateTheoremCSCCDPCCache} for the CSC scheme with $t=2$, for some $\boldsymbol{\alpha}$ and $\boldsymbol{\lambda}$ satisfying \eqref{AchievableRateTheoremCSCCDPCalphaiAndni}-\eqref{AchievableRateTheoremCSCCDPCalphaiAndnisumni}.

\begin{table}[!t]
\caption{Decoding the message with $x_2^{\lambda_{i} n}$ at TS $i$, for $i=1, ..., 4$}
\centering
\begin{tabular}{ |c|c| }
\hline
TS number & Sufficient conditions\\
\hline
\multirow{2.4}{*}{1} & \vtop{\hbox{\strut $R^{(2)}_{\{ 3\}} \le \lambda_1 C^{\bar{\alpha}_1 P}_{\sigma_2^2}$} \hbox{\strut $R^{(2)}_{\{ 2\}} \le \lambda_1 C^{\bar{\alpha}_1 P}_{\sigma_3^2}$}} \\
\hline
\multirow{2.4}{*}{2} & \vtop{\hbox{\strut $R^{(2)}_{\{ 4\}} \le \lambda_2 C^{\bar{\alpha}_2 P}_{\sigma_2^2}$} \hbox{\strut $R^{(2)}_{\{ 2\}} \le \lambda_2 C^{\bar{\alpha}_2 P}_{\sigma_4^2}$}}\\
\hline
\multirow{2.4}{*}{3} & \vtop{\hbox{\strut $R^{(2)}_{\{ 4\}} \le \lambda_3 C^{\bar{\alpha}_3 P}_{\sigma_3^2}$} \hbox{\strut $R^{(2)}_{\{ 3\}} \le \lambda_3 C^{\bar{\alpha}_3 P}_{\sigma_4^2}$}}\\
\hline
\multirow{2.4}{*}{4} & \vtop{\hbox{\strut $R^{(3)}_{\{ 4\}} \le \lambda_4 C^{\bar{\alpha}_4 P}_{\sigma_3^2}$} \hbox{\strut $R^{(3)}_{\{ 3\}} \le \lambda_4 C^{\bar{\alpha}_4 P}_{\sigma_4^2}$}}\\
\hline
\end{tabular}
\label{TableDeliveryExample1Codeword23}
\end{table}

\subsection{Joint Cache and Dirty Paper Coding (CDPC) Scheme}\label{CDPCScheme}

In the following, we investigate the delivery via TS $i$, for $i \in \left[1+g_{l-1}:g_{l}\right]$ and $l \in [K-t]$. The codebook of the transmission with the CDPC scheme is also generated from the linear superposition of two subcodebooks. The subcodebook $\mathcal{C}_2$, given in \eqref{CSCCodebooksDef2}, is generated from i.i.d. codewords $x_2^{\lambda_i n}$, each according to distribution $X_2 \sim \mathcal{N} \left( 0,\bar{\alpha}_i P \right)$, used to send the coded packet $W_{\mathcal{S}_{\mathcal{K}_{l+1},i-g_{l-1}}^{t}}^{(l+1)}$ of rate $R^{(l+1)}_{{\rm{XOR}},\mathcal{S}_{\mathcal{K}_{l+1},i-g_{l-1}}^{t}}$. By treating $x_2^{\lambda_i n}$ as interference for user $l$, knowing $X_2^{\lambda_{i}}$ non-causally at the server, subcodebook $\mathcal{C}_1$ is generated using dirty paper coding \cite{DirtyPaperCosta}. The auxiliary random variable with the dirty paper coding is set as $Q = X_1 + \tau X_2$, where $\tau = \alpha_i P/\left( \alpha_i P+\sigma_l^2 \right)$, and $X_1 \sim \mathcal{N} \left( 0,\alpha_i P \right)$ is independent of $X_2$. We extend the codebook generation, encoding, and decoding techniques of the Gelfand-Pinkser scheme for point-to-point transmission presented in the proof of \cite[Theorem 7.3]{AEGamalNetworkInfTheory} to the transmission in the $i$-th TS of the setting under consideration, for $i \in \left[1+g_{l-1}:g_{l}\right]$ and $l \in [K-t]$. We define a message tuple $\textbf{w}_{\mathcal{S}_{\mathcal{K}_1,i}^{t+1}} \buildrel \Delta \over = \left( w_{d_{k_1},\mathcal{S}_{\mathcal{K}_1,i}^{t+1} \backslash \{k_1\}}^{(1)}, \mbox{for $k_1 \in \mathcal{S}_{\mathcal{K}_1,i}^{t+1}$} \right)$, which concatenates $t+1$ messages $w_{d_{k_1},\mathcal{S}_{\mathcal{K}_1,i}^{t+1} \backslash \{k_1\}}^{(1)}$, $\forall k_1 \in \mathcal{S}_{\mathcal{K}_1,i}^{t+1}$, where $w_{d_{k_1},\mathcal{S}_{\mathcal{K}_1,i}^{t+1} \backslash \{k_1\}}^{(1)}$ is uniformly distributed over $\left[ 2^{n R^{(1)}_{\mathcal{S}_{\mathcal{K}_1,i}^{t+1}\backslash \{ k_1\}}} \right]$ and represents the message used to send subfile $W_{d_{k_1},\mathcal{S}_{\mathcal{K}_1,i}^{t+1} \backslash \{k_1\}}^{(1)}$, for $k_1 \in \mathcal{S}_{\mathcal{K}_1,i}^{t+1}$. For each realization of $\textbf{w}_{\mathcal{S}_{\mathcal{K}_1,i}^{t+1}}$, we generate a subcodebook $\mathcal{C}_1 \left( \textbf{w}_{\mathcal{S}_{\mathcal{K}_1,i}^{t+1}} \right)$ of $2^{\lambda_in \tilde{R} - n \sum\limits_{k_1 \in \mathcal{S}_{\mathcal{K}_1,i}^{t+1}} {R^{(1)}_{\mathcal{S}_{\mathcal{K}_1,i}^{t+1}\backslash \{ k_1\}}} }$ sequences $q^{\lambda_in} (m)$, for $m \in \left[ 2^{\lambda_i n \tilde{R} - n \sum\limits_{k_1 \in \mathcal{S}_{\mathcal{K}_1,i}^{t+1}} {R^{(1)}_{\mathcal{S}_{\mathcal{K}_1,i}^{t+1}\backslash \{ k_1\}}} } \right]$. Given $x_2^{\lambda_in}$, in order to send $t+1$ messages with message tuple $\textbf{w}_{\mathcal{S}_{\mathcal{K}_1,i}^{t+1}}$ jointly, we find a sequence $q^{\lambda_i n} (m) \in \mathcal{C}_1 \left( \textbf{w}_{\mathcal{S}_{\mathcal{K}_1,i}^{t+1}} \right)$ that is jointly typical with $x_2^{\lambda_i n}$ and represent the corresponding codeword, which is to be sent over the channel, by 
\begin{align}\label{Codewordx1CDPC}
x_1^{\lambda_i n} \left( \textbf{w}_{\mathcal{S}_{\mathcal{K}_1,i}^{t+1}},x_2^{\lambda_i n} \right).
\end{align}
The server then sends the following linearly superposed codeword over the Gaussian BC:
\begin{align}\label{CDPCServerTransmittedSignal}
    x_1^{\lambda_{i} n} \left( \left( W_{d_{k_1},\mathcal{S}_{\mathcal{K}_1,i}^{t+1} \backslash \{k_1\}}^{(1)}, \mbox{for $k_1 \in \mathcal{S}_{\mathcal{K}_1,i}^{t+1}$} \right), x_2^{\lambda_{i} n} \right) + x_2^{\lambda_{i} n} \left( W_{\mathcal{S}_{\mathcal{K}_{l+1},i-g_{l-1}}^{t}}^{(l+1)} \right).
\end{align}

User $k_2 \in {\mathcal{S}_{\mathcal{K}_{l+1},i-g_{l-1}}^{t}}$ first decodes the message with codeword $x_2^{\lambda_{i}n}$, while considering $x_1^{\lambda_{i}n}$ as noise, which by the same analysis as the CSC scheme, one can obtain the necessary condition for a successful decoding as follows:
\begin{align}\label{CDPCRecoverUsersFirstLayerDecoding}
    & R^{(l+1)}_{\mathcal{S}_{\mathcal{K}_{l+1},i-g_{l-1}}^{t} \backslash \{k_2\}} \le \lambda_i C^{\bar{\alpha}_i P}_{\alpha_i P+\sigma_{k_2}^2}, \quad \mbox{for $i \in \left[1+g_{l-1}:g_{l}\right]$ and $l \in [K-t]$}. 
\end{align}
On the other hand, to decode $t+1$ messages with $x_1^{\lambda_i n}$, upon receiving $y_{k_1}^{\lambda_i n}$, user $k_1$, for $k_1 \in {\mathcal{S}_{\mathcal{K}_{1}, i}^{t+1}}$, declares that $t+1$ messages $\hat{w}_{d_{\tilde{k}_1},\mathcal{S}_{\mathcal{K}_1,i}^{t+1} \backslash \{\tilde{k}_1\}}^{(1)} \in \left[ 2^{n R^{(1)}_{\mathcal{S}_{\mathcal{K}_1,i}^{t+1}\backslash \{ \tilde{k}_1\}}} \right]$, $\forall \tilde{k}_1 \in {\mathcal{S}_{\mathcal{K}_{1}, i}^{t+1}}$, are sent if $\hat{\textbf{w}}_{\mathcal{S}_{\mathcal{K}_1,i}^{t+1}} \buildrel \Delta \over = \left( \hat{w}_{d_{\tilde{k}_1},\mathcal{S}_{\mathcal{K}_1,i}^{t+1} \backslash \{\tilde{k}_1\}}^{(1)}, \mbox{for $\tilde{k}_1 \in \mathcal{S}_{\mathcal{K}_1,i}^{t+1}$} \right)$ is the unique message tuple such that $q^{\lambda_i n} (m)$ and $y_{k_1}^{\lambda_i n}$ are jointly typical, for some $m \in \mathcal{C}_1 \left( \hat{\textbf{w}}_{\mathcal{S}_{\mathcal{K}_1,i}^{t+1}} \right)$, where message ${w}_{d_{\tilde{k}_1},\mathcal{S}_{\mathcal{K}_1,i}^{t+1} \backslash \{\tilde{k}_1\}}^{(1)}$ is decoded as $\hat{w}_{d_{\tilde{k}_1},\mathcal{S}_{\mathcal{K}_1,i}^{t+1} \backslash \{\tilde{k}_1\}}^{(1)}$, for $\tilde{k}_1 \in {\mathcal{S}_{\mathcal{K}_{1}, i}^{t+1}}$.\footnote{For ease of notation, we drop the dependency of channel outputs ${Y^n_{1}(\textbf{W})}, \ldots, {Y^n_{K}(\textbf{W})}$ on $\textbf{W}$.} Here we note again that, for $l \in [K-t]$,
\begin{align}\label{CDPCLabelSubsetsInsideeachGroupK1Kk}
    {\mathcal{S}_{\mathcal{K}_{1}, i}^{t+1}} = \left\{{\mathcal{S}_{\mathcal{K}_{l+1}, i-g_{l-1}}^{t}},  l \right\}, \quad \mbox{for $i \in \left[1+g_{l-1}:g_{l}\right]$}. 
\end{align}
We assume without loss of generality that the message tuple $\textbf{w}_{\mathcal{S}_{\mathcal{K}_1,i}^{t+1}} = \left(1,\ldots,1 \right)$ is sent with $x_1^{\lambda_i n}$. The decoder at user $l \in {\mathcal{S}_{\mathcal{K}_{1}, i}^{t+1}}$ makes an error, if one or both of the following events occur:
\begin{subequations}
\label{CDPCUser1ErrorEvent}
\begin{align}\label{CDPCUser1ErrorEvent1}
    \mathcal{E}^{1}_{l} = & \left\{\mbox{$Q^{\lambda_i n}(m)$ and $X_2^{\lambda_i n}$ are not jointly typical, $\forall Q^{\lambda_i n}(m) \in \mathcal{C}_1 (1,\ldots,1)$} \right\}, \\
    \mathcal{E}^{2}_{l} = & \left\{\mbox{$Q^{\lambda_i n}(m)$ and $Y_l^{\lambda_i n}$ are jointly typical, for some $Q^{\lambda_i n}(m) \notin \mathcal{C}_1 (1,\ldots,1)$} \right\}.\label{CDPCUser1ErrorEvent2}
\end{align}
\end{subequations}
According to \cite[Lemma 3.3]{AEGamalNetworkInfTheory}, $\Pr \left\{ \mathcal{E}^{1}_{l}  \right\}$ tends to zero, if, for $n$ large enough, 
\begin{align}\label{CDPCUser1ErrorEvent1TendsZero}
\lambda_i \tilde{R} - \sum\limits_{k_1 \in \mathcal{S}_{\mathcal{K}_1,i}^{t+1}} {R^{(1)}_{\mathcal{S}_{\mathcal{K}_1,i}^{t+1}\backslash \{ k_1\}}} \ge \lambda_i I \left( Q;X_2 \right). 
\end{align}
Furthermore, since user $l \in \mathcal{S}_{\mathcal{K}_1,i}^{t+1}$ has access to all $t$ messages $W_{d_{k},\mathcal{S}_{\mathcal{K}_1,i}^{t+1} \backslash \{k\}}^{(1)}$, $\forall k \in \mathcal{S}_{\mathcal{K}_1,i}^{t+1} \backslash \{ l \}$, in its cache, it knows that $w_{d_{k},\mathcal{S}_{\mathcal{K}_1,i}^{t+1} \backslash \{k\}}^{(1)} = 1$, $\forall k \in \mathcal{S}_{\mathcal{K}_1,i}^{t+1} \backslash \{ l \}$, and $\Pr \left\{ \mathcal{E}^{2}_{l}  \right\}$ tends to zero, if, for $n$ large enough, 
\begin{align}\label{CDPCUser1ErrorEvent2TendsZero}
\lambda_i \tilde{R} - \sum\limits_{k_1 \in \mathcal{S}_{\mathcal{K}_1,i}^{t+1} \backslash \{ l \}} {R^{(1)}_{\mathcal{S}_{\mathcal{K}_1,i}^{t+1}\backslash \{ l\}}} \le \lambda_i I \left( Q;Y_l \right), 
\end{align}
where $Y_k = X_1+X_2+Z_k$, where $Z_k \sim \mathcal{N} \left( 0,\sigma_k^2 \right)$, for $k \in [K]$. Combining \eqref{CDPCUser1ErrorEvent1TendsZero} and \eqref{CDPCUser1ErrorEvent2TendsZero}, we obtain that user $l \in \mathcal{S}_{\mathcal{K}_1,i}^{t+1}$ decodes the message with $x_1^{\lambda_{i} n}$ successfully, if 
\begin{align}\label{CDPCRecoverUsersFirstLayerDecodingGeneralFirstUser}
    R^{(1)}_{\mathcal{S}_{\mathcal{K}_1,i}^{t+1} \backslash \{l\}} \le \lambda_i \left( I \left( Q;Y_l \right) - I \left( Q;X_2 \right) \right), 
\end{align}
which is equivalent to
\begin{align}\label{CDPCRecoverUsersHigherLayerDecodingFirstUser}
    R^{(1)}_{\mathcal{S}_{\mathcal{K}_1,i}^{t+1} \backslash \{l\}} \le \lambda_i C^{\alpha_iP}_{\sigma_{l}^2}. 
\end{align}
Now we investigate the sufficient conditions for which users in ${\mathcal{S}_{\mathcal{K}_{1}, i}^{t+1}} \backslash \{ l \} = {\mathcal{S}_{\mathcal{K}_{l+1}, i-g_{l-1}}^{t}}$ can decode the message with $x_1^{\lambda_i n}$. We note that having the conditions in \eqref{CDPCRecoverUsersFirstLayerDecoding} satisfied, each user in ${\mathcal{S}_{\mathcal{K}_{l+1}, i-g_{l-1}}^{t}}$ can decode the message with $x_2^{\lambda_i n}$. The decoder at user $k$, for $k \in {\mathcal{S}_{\mathcal{K}_{l+1}, i-g_{l-1}}^{t}}$, makes an error, if one or both of the following events occur:
\begin{subequations}
\label{CDPCHigherUsersErrorEvent}
\begin{align}\label{CDPCHigherUsersErrorEvent1}
    \mathcal{E}^{1}_{k} = & \left\{\mbox{$Q^{\lambda_i n}(m)$ and $X_2^{\lambda_in}$ are not jointly typical, $\forall Q^{\lambda_in}(m) \in \mathcal{C}_1 (1,\ldots,1)$} \right\}, \\
    \mathcal{E}^{2}_{k} = & \left\{\mbox{$Q^{\lambda_in}(m)$ and $\left( Y_{k}^{\lambda_in} , X_2^{\lambda_in} \right)$ are jointly typical, for some $Q^{\lambda_in}(m) \notin \mathcal{C}_1 (1,\ldots,1)$} \right\}.\label{CDPCHigherUsersErrorEvent2}
\end{align}
\end{subequations}
$\Pr \left\{ \mathcal{E}^{1}_{k}  \right\}$ tends to zero, if, for $n$ large enough,
\begin{align}\label{CDPCHigherUsersErrorEvent1TendsZero}
\lambda_i \tilde{R} - \sum\limits_{k_1 \in \mathcal{S}_{\mathcal{K}_1,i}^{t+1}} {R^{(1)}_{\mathcal{S}_{\mathcal{K}_1,i}^{t+1}\backslash \{ k_1\}}} \ge \lambda_i I \left( Q;X_2 \right), \quad \mbox{for $k \in {\mathcal{S}_{\mathcal{K}_{l+1}, i-g_{l-1}}^{t}}$}. 
\end{align}
Furthermore, since user $k$, for $k \in {\mathcal{S}_{\mathcal{K}_{l+1}, i-g_{l-1}}^{t}}$ has access to all $t$ messages $W_{d_{k_1},\mathcal{S}_{\mathcal{K}_1,i}^{t+1} \backslash \{k_1\}}^{(1)}$, $\forall k_1 \in \mathcal{S}_{\mathcal{K}_1,i}^{t+1} \backslash \{ k \}$, in its cache, it knows that $w_{d_{k_1},\mathcal{S}_{\mathcal{K}_1,i}^{t+1} \backslash \{k_1\}}^{(1)} = 1$, $\forall k_1 \in \mathcal{S}_{\mathcal{K}_1,i}^{t+1} \backslash \{ k \}$, and $\Pr \left\{ \mathcal{E}^{2}_{k}  \right\}$ tends to zero, if, for $n$ large enough, 
\begin{align}\label{CDPCHigherUsersErrorEvent2TendsZero}
\lambda_i \tilde{R} - \sum\limits_{k_1 \in \mathcal{S}_{\mathcal{K}_1,i}^{t+1} \backslash \{ k \}} {R^{(1)}_{\mathcal{S}_{\mathcal{K}_1,i}^{t+1}\backslash \{ k_1\}}} \le \lambda_i I \left( Q;Y_l,X_2 \right). 
\end{align}
Combining \eqref{CDPCHigherUsersErrorEvent1TendsZero} and \eqref{CDPCHigherUsersErrorEvent2TendsZero}, we obtain that user $k$, for $k \in {\mathcal{S}_{\mathcal{K}_{l+1}, i-g_{l-1}}^{t}}$ decodes the message with $x_1^{\lambda_{i}n}$ successfully, if 
\begin{align}\label{CDPCRecoverUsersFirstLayerDecodingGeneralHigherUsers}
    R^{(1)}_{\mathcal{S}_{\mathcal{K}_1,i}^{t+1} \backslash \{{k}\}} \le \lambda_i \left( I \left( Q;Y_{k},X_2 \right) - I \left( Q;X_2 \right) \right) = \lambda_i I \left( Q;Y_{k} \left| X_2 \right. \right), 
\end{align}   
which leads to
\begin{align}\label{CDPCRecoverUsersHigherLayerDecodingHigherUsers}
    R^{(1)}_{\mathcal{S}_{\mathcal{K}_1,i}^{t+1} \backslash \{{k}\}} \le \lambda_i C^{\alpha_i P}_{\sigma_{k}^2}. 
\end{align}
By combining the conditions in \eqref{CDPCRecoverUsersHigherLayerDecodingFirstUser} and \eqref{CDPCRecoverUsersHigherLayerDecodingHigherUsers}, we conclude that, at TS $i$, user $k_1$, for $k_1 \in \mathcal{S}_{\mathcal{K}_1,i}^{t+1}$, can decode the message with $x_1^{\lambda_{i}n}$ successfully, if
\begin{align}\label{CDPCRecoverUsersHigherLayerDecoding}
    R^{(1)}_{\mathcal{S}_{\mathcal{K}_1,i}^{t+1} \backslash \{{k_1}\}} \le \lambda_i C^{\alpha_i P}_{\sigma_{k_1}^2}, \quad &\mbox{for $i \in \left[1+g_{l-1}:g_{l}\right]$ and $l \in [K-t]$}. 
\end{align}

Having the conditions in \eqref{CDPCRecoverUsersFirstLayerDecoding} and \eqref{CDPCRecoverUsersHigherLayerDecoding} satisfied, the achievability of the rate tuple for the corresponding total cache capacity $M$ presented in Theorem \ref{AchievablePairsCSCSchemeTheorem} for the CDPC scheme is proved.

\section{Outer Bound}\label{OuterBound} 
In the following, we develop an outer bound on the capacity region $\mathcal{C} (P,M)$ constrained to uncoded caching in the placement phase.

\begin{theorem}\label{OuterBoundTheorem}
Consider the system described in Section \ref{SystemModel} with average power $P$, where user $k$ has a cache capacity of $M_k$, $k \in [K]$. If the placement phase is constrained to uncoded caching, for any non-empty subset $\mathcal{G} \subset [K]$, we have, for $k=1, \ldots, \left| \mathcal{G} \right|$, 
\begin{align}\label{OuterBoundRk}
R_{\pi_{\mathcal{G}}(k)} \le C_{\sum\limits_{i=k+1}^{\left| \mathcal{G} \right|} {\eta^{\mathcal G}_{\pi_{\mathcal{G}}(i)} P} + \sigma_{\pi_{\mathcal{G}}(k)}^2}^{\eta^{\mathcal G}_{\pi_{\mathcal{G}}(k)} P} + \frac{1}{N} \sum\limits_{i=1}^{k} M_{\pi_{\mathcal{G}}(i)}, 
\end{align}
for some non-negative coefficients $\eta^{\mathcal G}_{\pi_{\mathcal{G}}(1)}, \ldots, \eta^{\mathcal G}_{\pi_{\mathcal{G}}(\left| \mathcal G \right|)}$, such that $\sum\nolimits_{i=1}^{\left| \mathcal G \right|} \eta^{\mathcal G}_{\pi_{\mathcal{G}}(i)} =1$, where $\pi_{\cal G}$ is a permutation of the elements of ${\cal G}$, such that $\sigma_{\pi_{\mathcal G}(1)}^2 \ge \sigma_{\pi_{\mathcal G}(2)}^2 \ge \cdots \ge \sigma_{\pi_{\mathcal G}(\left| \mathcal G \right|)}^2$.
\end{theorem}

\begin{proof}
For ease of presentation, we prove the outer bound for $\mathcal{G} = [K]$, and the proof of general case follows similarly. By an abuse of the notation, for a demand vector $\textbf{d}$ with all different entries and noise variances $\boldsymbol{\sigma}$ in the delivery phase, we denote the channel input, generated by function $\psi_{\boldsymbol{\sigma},\textbf{d}}$, by $X^n_{\textbf{d}}$, and the channel output at user $k$ by $Y^n_{\textbf{d},k}$, where
\begin{equation}\label{ChannelModelAppendix} 
{Y^n_{\textbf{d},k}} = {X^n_{\textbf{d}}} + {Z^n_{k}}, \quad \mbox{for $k \in [K]$}. 
\end{equation}

\begin{lemma}\label{LemmaBCProofOuterBound}
Let $\left( R_1, \dots, R_K \right)$ be an achievable rate tuple. For a demand vector $\emph{\textbf{d}} = \left( d_1,\dots,d_K\right)$ with all distinct entries, there exist random variables $X_{\emph{\textbf{d}}}$, $Y_{\emph{\textbf{d}}.1}, \dots, Y_{\emph{\textbf{d}},K}$, and $\left\{ V_{\emph{\textbf{d}},1}, \dots, V_{\emph{\textbf{d}},K-1} \right\} $, where
\begin{equation}\label{MarkovChainOuterBound}
    V_{\emph{\textbf{d}},1} \to \cdots \to V_{\emph{\textbf{d}},K-1} \to X_{\emph{\textbf{d}}} \to Y_{\emph{\textbf{d}},K}  \to \cdots \to Y_{\emph{\textbf{d}},1}
\end{equation}
forms a Markov chain, and satisfy
\begin{subequations}
\label{ShirinAppendixLowerBound}
\begin{align}\label{ShirinAppendixLowerBound1}
R_1 - {\varepsilon} \le & I\left( {V}_{\emph{\textbf{d}},1};Y_{\emph{\textbf{d}},1} \right) + \frac{1}{n}I\left( {{W^{(1)}_{d_1}};{U_{1}}} \right),\\
\label{ShirinAppendixLowerBound2}
R_k - {\varepsilon} \le & I\left( {V}_{\emph{\textbf{d}},k}; Y_{\emph{\textbf{d}},k} \left| {V}_{\emph{\textbf{d}},k-1} \right. \right) + \frac{1}{n}I\left( \bigcup\limits_{l=1}^{k} W_{d_k}^{(l)};U_{1}, \dots, U_{k} \left| \bigcup\limits_{m=1}^{k-1} \bigcup\limits_{l=1}^{m} W_{d_m}^{(l)} \right. \right), \; \forall k \in [2:K-1],\\
\label{ShirinAppendixLowerBound3}
R_K - {\varepsilon} \le & I\left( X_{\emph{\textbf{d}}}; Y_{\emph{\textbf{d}},K} \left| {V}_{\emph{\textbf{d}},K-1} \right. \right) + \frac{1}{n}I\left( \bigcup\limits_{l=1}^{K} W_{d_k}^{(l)};U_{1}, \dots, U_{K} \left| \bigcup\limits_{m=1}^{K-1} \bigcup\limits_{l=1}^{m} W_{d_m}^{(l)} \right. \right),
\end{align}
\end{subequations}
where ${\varepsilon} >0$ tends to zero as $n \to \infty$. 
\end{lemma}
\begin{proof}
See Appendix \ref{ProffOfLemmaBC}. 
\end{proof}
Assuming $N \ge K$, let $\mathcal{D}_k$ be the set of all $\binom{N}{k} k!$ $k$-dimensional vectors, where all entries of each vector are distinct, and each entry of every vector takes a value in $[N]$, for $k \in [K]$. We note that $\mathcal{D}_K$ is the set of all demand vectors, each with all different entries. By averaging over all demand vectors with different entries, we can obtain from Lemma \ref{LemmaBCProofOuterBound} that
\begin{subequations}
\label{ProofOuterBoundR1}
\begin{align}\label{ProofOuterBoundR11}
R_1 - {\varepsilon} \le & I\left( {V}_{\textbf{d},1};Y_{\textbf{d},1} \right) + \frac{1}{\binom{N}{K}{K!}} \sum\limits_{\textbf{d} \in \mathcal{D}_K} \frac{1}{n}I\left( {{W^{(1)}_{d_1}};{U_{1}}} \right)\\
\label{ProofOuterBoundR12}
= & I\left( {V}_{\textbf{d},1};Y_{\textbf{d},1} \right) + \frac{1}{\binom{N}{K}{K!}} \binom{N-1}{K-1}(K-1)! \sum\nolimits_{j=1}^{N} \frac{1}{n}I\left( {{W^{(1)}_{j}};{U_{1}}} \right)\\
\label{ProofOuterBoundR13}
= & I\left( {V}_{\textbf{d},1};Y_{\textbf{d},1} \right) + \frac{1}{N} \sum\nolimits_{j=1}^{N} \frac{1}{n}I\left( {{W^{(1)}_{j}};{U_{1}}} \right)\\
\label{ProofOuterBoundR14}
\le & I\left( {V}_{\textbf{d},1};Y_{\textbf{d},1} \right) + \frac{1}{nN} I\left( \textbf{W}^{(1)};U_{1} \right)\\
\label{ProofOuterBoundR15}
\le & I\left( {V}_{\textbf{d},1};Y_{\textbf{d},1} \right) + \frac{M_1}{N},
\end{align}
\end{subequations}
where \eqref{ProofOuterBoundR14} follows from the independence of the files, and, for $k \in [2:K]$, 
\begin{subequations}
\label{ProofOuterBoundRk}
\begin{align}\label{ProofOuterBoundRk1}
R_k - {\varepsilon} \le & I\left( {V}_{\textbf{d},k}; Y_{\textbf{d},k} \left| {V}_{\textbf{d},k-1} \right. \right) + \frac{1}{\binom{N}{K}{K!}} \sum\limits_{\textbf{d} \in \mathcal{D}_K} \frac{1}{n}I\left( \bigcup\limits_{l=1}^{k} W_{d_k}^{(l)};U_{1}, \dots, U_{k} \left| \bigcup\limits_{m=1}^{k-1} \bigcup\limits_{l=1}^{m} W_{d_m}^{(l)} \right. \right)\\
\label{ProofOuterBoundRk2}
= & I\left( {V}_{\textbf{d},k}; Y_{\textbf{d},k} \left| {V}_{\textbf{d},k-1} \right. \right) + \nonumber\\
& \qquad \; \; \frac{1}{\binom{N}{K}{K!}} \sum\limits_{\tilde{\textbf{d}} \in \mathcal{D}_{k-1}} \sum\limits_{\textbf{d} \in \mathcal{D}_K: \left( d_1, ..., d_{k-1} \right) = \tilde{\textbf{d}}} \frac{1}{n}I\left( \bigcup\limits_{l=1}^{k} W_{d_k}^{(l)};U_{1}, \dots, U_{k} \left| \bigcup\limits_{m=1}^{k-1} \bigcup\limits_{l=1}^{m} W_{d_m}^{(l)} \right. \right)\\
\label{ProofOuterBoundRk3}
= & I\left( {V}_{\textbf{d},k}; Y_{\textbf{d},k} \left| {V}_{\textbf{d},k-1} \right. \right) + \frac{1}{\binom{N}{K}{K!}} \sum\limits_{\tilde{\textbf{d}} \in \mathcal{D}_{k-1}} \sum\limits_{\textbf{d} \in \mathcal{D}_K: \left( d_1, ..., d_{k-1} \right) = \tilde{\textbf{d}}} \frac{1}{n}I\left( \bigcup\limits_{l=1}^{k} W_{d_k}^{(l)};U_{1}, \dots, U_{k} \right)\\
\label{ProofOuterBoundRk4}
= & I\left( {V}_{\textbf{d},k}; Y_{\textbf{d},k} \left| {V}_{\textbf{d},k-1} \right. \right) +\nonumber\\
&\frac{1}{\binom{N}{K}{K!}} \sum\limits_{\tilde{\textbf{d}} \in \mathcal{D}_{k-1}} \sum\limits_{j \in [N] \backslash \left\{ \tilde{d}_1, ..., \tilde{d}_{k-1} \right\}} \frac{1}{n}I\left( \bigcup\limits_{l=1}^{k} W_{j}^{(l)};U_{1}, \dots, U_{k} \right) \binom{N-k}{K-k}(K-k)!\\
\label{ProofOuterBoundRk5}
= & I\left( {V}_{\textbf{d},k}; Y_{\textbf{d},k} \left| {V}_{\textbf{d},k-1} \right. \right) + \nonumber\\
& \frac{1}{\binom{N}{K}{K!}} \sum\limits_{j = 1}^{N} \frac{1}{n}I\left( \bigcup\limits_{l=1}^{k} W_{j}^{(l)};U_{1}, \dots, U_{k} \right) \binom{N-k}{K-k}(K-k)! \binom{N-1}{k-1}(k-1)!
\end{align}

\begin{align}
\label{ProofOuterBoundRk6}
= & I\left( {V}_{\textbf{d},k}; Y_{\textbf{d},k} \left| {V}_{\textbf{d},k-1} \right. \right) + \frac{1}{N} \sum\limits_{j = 1}^{N} \frac{1}{n}I\left( \bigcup\limits_{l=1}^{k} W_{j}^{(l)};U_{1}, \dots, U_{k} \right)\\
\label{ProofOuterBoundRk7}
\le & I\left( {V}_{\textbf{d},k}; Y_{\textbf{d},k} \left| {V}_{\textbf{d},k-1} \right. \right) + \frac{1}{nN}I\left( \bigcup\limits_{l=1}^{k} \textbf{W}^{(l)};U_{1}, \dots, U_{k} \right)\\
\label{ProofOuterBoundRk8}
\le & I\left( {V}_{\textbf{d},k}; Y_{\textbf{d},k} \left| {V}_{\textbf{d},k-1} \right. \right) + \frac{1}{N} \sum\limits_{i=1}^{k} M_k,
\end{align}
\end{subequations}
where \eqref{ProofOuterBoundRk3} follows from the assumption of uncoded caching and the independence of the files, $\tilde{d}_i$ in \eqref{ProofOuterBoundRk4}, for $i \in [k-1]$, returns the $i$-th element of vector $\tilde{\textbf{d}}$, \eqref{ProofOuterBoundRk7} follows from the the independence of the files, and we define $V_{\textbf{d},K} \buildrel \Delta \over = X$. For the Gaussian channel, we have \cite{BergmansCapacityDegradeBC}
\begin{equation}\label{LowerBoundAppendixBergman}
    I\left( V_{\textbf{d},k}; Y_{\textbf{d},k} \left| {V}_{\textbf{d},k-1} \right. \right) \le C_{\sum\limits_{i=k+1}^{K} {\eta_{i} P} + \sigma_{k}^2}^{\eta_{k} P} , \quad \mbox{for $k \in [K]$},
\end{equation}
for some non-negative coefficients $\eta_{1}, \ldots, \eta_{K}$, such that $\sum\nolimits_{i=1}^{K} \eta_{i} =1$, where we set ${V}_{\textbf{d},0} \buildrel \Delta \over = 0$. This completes the proof of Theorem \ref{OuterBoundTheorem} for $\mathcal{G} = [K]$. The proof can be extended to the general case by taking similar steps.  

\end{proof}

The outer bound on the capacity region can be found by considering the bound in \eqref{OuterBoundRk}, which provides an upper bound on the rate delivered to user $\pi_{\mathcal G}(k)$, for $k \in \mathcal G$, for all non-empty subsets $\mathcal{G} \subset [K]$ and all possible cache allocations $M_1, \ldots, M_K$, such that $\sum\nolimits_{k=1}^{K} M_k=M$. The convex hull of these tuples, calculated through \eqref{OuterBoundRk}, $\forall \mathcal{G} \subset [K]$ and all possible cache allocations with a total cache capacity $M$, also provides an outer bound on the capacity region. As a result, for given non-negative coefficients $w_1, \dots, w_K$, rate tuple $\left( R_1, \ldots, R_K \right)$ is on the boundary surface of the outer bound, if $R_{1}, \dots, R_{K}$ is a solution of the following problem:
\begin{subequations}
\label{DetailedOptimizationOuterBound}
\begin{align}\label{DetailedOptimizationOuterBound1}
&\mathop {\max }\limits_{\boldsymbol{\eta},M_1, \ldots, M_K,R_1, \ldots, R_K} \sum\limits_{i=1}^{K} w_i R_i, \nonumber\\
& \mbox{subject to $R_{\pi_{\mathcal{G}}(k)} \le C_{\sum\limits_{i=k+1}^{\left| \mathcal{G} \right|} {\eta^{\mathcal G}_{\pi_{\mathcal{G}}(i)} P} + \sigma_{\pi_{\mathcal{G}}(k)}^2}^{\eta^{\mathcal G}_{\pi_{\mathcal{G}}(k)} P} + \frac{1}{N} \sum\limits_{i=1}^{k} M_{\pi_{\mathcal{G}}(i)}$}, \; \forall k \in \mathcal{G}, \forall \mathcal G \subset [K]; \left| \mathcal{G} \right| \ne 0,\nonumber \\
& \qquad \qquad \; \sum\nolimits_{i=1}^{\left| \mathcal G \right|} \eta^{\mathcal G}_{\pi_{\mathcal{G}}(i)} =1, \; \forall \mathcal G \subset [K]; \left| \mathcal{G} \right| \ne 0, \nonumber \\
& \qquad \qquad \; \sum\limits_{k=1}^{K} M_k = M,
\end{align}
where 
\begin{align}\label{DetailedOptimizationOuterBoundVecEtaDef}
\boldsymbol{\eta} \buildrel \Delta \over =  \bigcup\limits_{\mathcal G \subset [K]: \left| \mathcal{G} \right| \ne 0} \eta^{\mathcal G}_{\pi_{\mathcal{G}}(1)}, \ldots, \eta^{\mathcal G}_{\pi_{\mathcal{G}}(\left| \mathcal G \right|)}.
\end{align}
\end{subequations}

\begin{remark}\label{RemLoosenessOuterBound}
The outer bound is not tight in general, particularly when the channel qualities are more skewed. This is due to the nature of the underlying model, where cache allocation is allowed, and the capacity is characterized as a function of the available total cache capacity, whereas the outer bound is specified for a particular cache allocation. Moreover, unlike the model studied in \cite{ShirinWiggerYenerCacheAssingment}, the asymmetry due to different rate delivery to different users increases the gap between the outer bound and the achievable schemes.
\end{remark}

\section{Numerical Results}\label{Comparison} 
In this section, we compare the achievable rate regions of the CTDC, CSC, and CDPC schemes for a caching system with $K=N=4$. We set the average power constraint to $P=2$, and the noise variance at user $k$ is assumed to be $\sigma_k^2 = 5-k$, for $k \in [4]$. We assume a total cache capacity of $M=2.5$.

\begin{figure}[!t]
\centering
\includegraphics[scale=0.5]{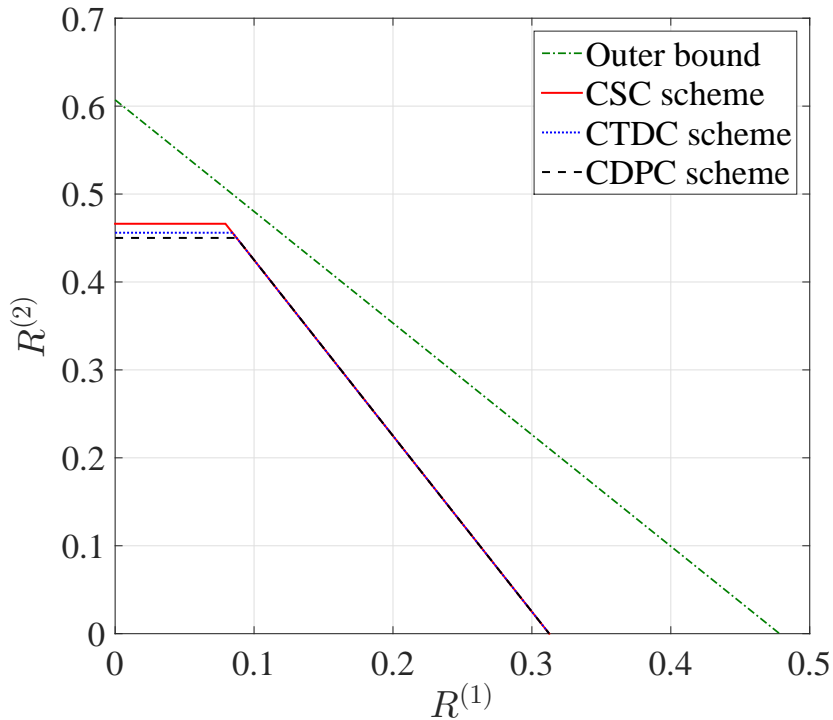}
\caption{Achievable rate pair $\left( R^{(1)},R^{(2)} \right)$ for a caching system with $K=N=4$, and $M=2.5$, where $R^{(3)}=0$, $t=2$ and $t_1=2$, $t_2=t_3=1$, and $t_4=0$. The noise variance at user $k$ is $\sigma_k^2 = 5-k$, for $k=1, ..., 4$, and we set $P=2$.}
\label{K4_N4_R1R2}
\end{figure}

We evaluate the performance in terms of the rate of different layers of the files, i.e., $R^{(1)}, \ldots,$ $R^{(K)}$, where $R_k = \sum\nolimits_{l=1}^{k} R^{(l)}$, for $k \in [K]$. We examine the performance of the CSC and CDPC schemes for $t=2$. Thus, the achievable rate tuple $\left( R_1, R_2, R_3, R_4 \right)$ presented in Theorem \ref{AchievablePairsCSCSchemeTheorem} can be achieved by the CSC and CDPC schemes, for $r_1=0$, $r_2=1$, and $r_1=1$, $r_2=0$, respectively, where $R_4=R_3$ since $R^{(4)}=0$. The boundary surface of the rate region achieved by the CSC and CDPC schemes are computed through the optimization problem given in \eqref{VectorRsuperlDef}. For the fairness of the comparison, we consider caching factors $t_1 = 2$, $t_2=t_3 =1$, and $t_4=0$.  
The boundary of the rate region achievable by CTDC can be calculated by the optimization problem in \eqref{DetailedOptimizationCTDC}, where, in order to have a fair comparison, we set $\lambda^{(4)}=0$ leading to $R^{(4)}=0$ and $R_4 = R_3$.

\begin{figure}[!t]
\centering
\includegraphics[scale=0.5]{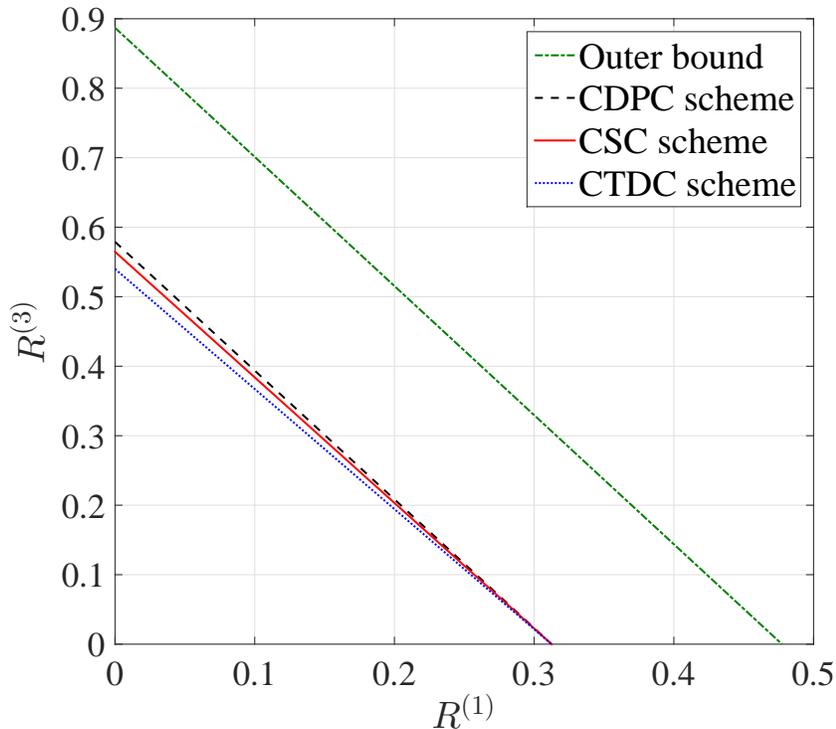}
\caption{Achievable rate pair $\left( R^{(1)},R^{(3)} \right)$ for a caching system with $K=N=4$, and $M=2.5$, where $R^{(2)} =0$, and $t=2$ and $t_1=2$, $t_2=t_3=1$, and $t_4=0$. The noise variance at user $k$ is $\sigma_k^2 = 5-k$, for $k=1, ..., 4$, and we set $P=2$.}
\label{K4_N4_R1R3}
\end{figure}

We investigate the convex hull of the achievable rate tuples calculated by the optimization problem corresponding to each of the CTDC, CSC, and CDPC schemes. Since the presentation of the three-dimensional rate region together with the outer bound does not provide a clear picture, here we fix one of the rates $R^{(1)}$, $R^{(2)}$ and $R^{(3)}$ and present the rate region on the two-dimensional planes corresponding to the other two rates. Two-dimensional plane of $\left( R^{(1)},R^{(2)} \right)$, $\left( R^{(1)},R^{(3)} \right)$ and $\left( R^{(2)},R^{(3)} \right)$ for $R^{(3)}=0$, $R^{(2)}=0$ and $R^{(1)}=0$ are illustrated in in Figures \ref{K4_N4_R1R2}, \ref{K4_N4_R1R3} and \ref{K4_N4_R2R3}, respectively, together with the outer bound presented in Theorem \ref{OuterBoundTheorem}. As it can be seen from the figures, for relatively small values of $R^{(1)}$, the CSC and CTDC schemes achieve higher values of $R^{(2)}$, while the CSC scheme outperforms the latter. For higher values of $R^{(1)}$, the improvement of the CSC scheme over CTDC and CDPC is negligible. For a fixed $R^{(1)}$ value, CDPC achieves higher values of $R^{(3)}$ compared to the other two achievable schemes, and CSC outperforms CTDC. On the other hand, given a relatively small value of $R^{(2)}$, CDPC improves upon the CSC and CTDC in terms of the achievable rate $R^{(3)}$, and CSC achieves higher values of $R^{(3)}$ than CTDC. As mentioned in Remark \ref{RemLoosenessOuterBound}, the outer bound is not tight in general; however, for any achievable rate tuple $\left( R_1, \ldots, R_4 \right)$, which is achieved with a specific cache allocation $M_1, \ldots, M_4$, the outer bound specialized to this cache allocation would be tighter.

\begin{figure}[!t]
\centering
\includegraphics[scale=0.5]{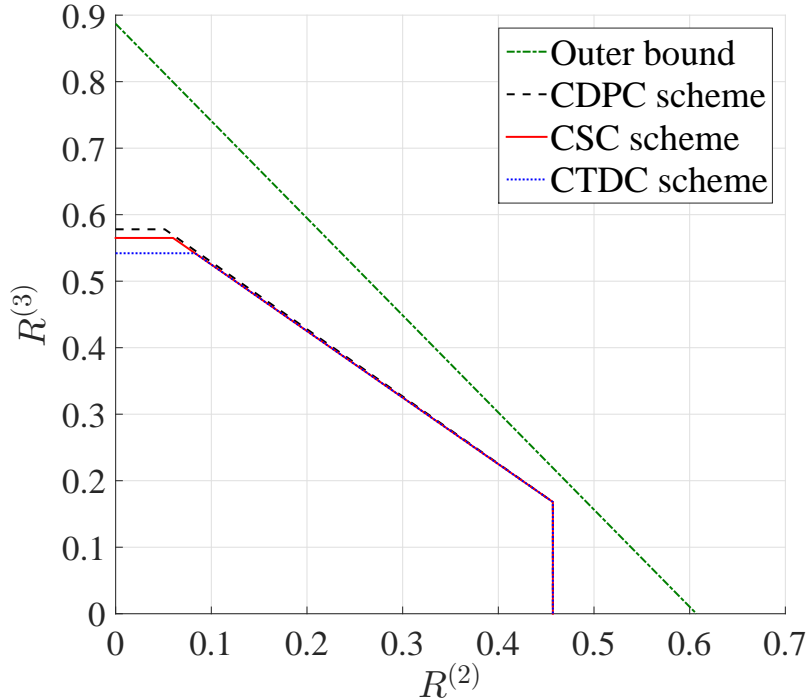}
\caption{Achievable rate pair $\left( R^{(2)},R^{(3)} \right)$ for a caching system with $K=N=4$, and $M=2.5$, where $R^{(1)}=0$, and $t=2$ and $t_1=2$, $t_2=t_3=1$, and $t_4=0$. The noise variance at user $k$ is $\sigma_k^2 = 5-k$, for $k=1, ..., 4$, and we set $P=2$.}
\label{K4_N4_R2R3}
\end{figure}

\section{Conclusions}\label{Conc}
We have studied cache-aided content delivery over a Gaussian BC, where each user is allowed to demand a file at a distinct rate. To model this asymmetry, we have assumed that the files are encoded into $K$ layers corresponding to $K$ users in the system, such that the $k$-th worst user is delivered only the $k$ layers of its demand, $k \in [K]$. We have considered a centralized placement phase, where the server knows the channel qualities of the links in the delivery phase in addition to the identity of the users. By allowing the users to have different cache capacities, we have defined the capacity region for a total cache capacity. We designed a placement phase through cache allocation across the users and the files' layers to maximize the rates allocated to different layers. We have proposed three achievable schemes, which deliver coded multicast packets, generated thanks to the contents carefully cached during the placement phase, through different channel coding techniques over the Gaussian BC. Although the coded multicast packets are intended for a set of users with distinct link capacities, channel coding techniques can be employed to deliver requested files such that the users with better channels achieve higher rates. We have also developed an outer bound on the capacity region assuming uncoded caching. We are currently working to reduce the gap between the inner and outer bounds.

\appendices

\section{Proof of Lemma \ref{LemmaBCProofOuterBound}}\label{ProffOfLemmaBC}

We follow the same steps as in \cite[Lemma 14]{ShirinWiggerYenerCacheAssingment}, but for multi-layer massages. Given a demand vector $\textbf{d}$ with all different entries and $\boldsymbol{\sigma}$ in the delivery phase, consider an achievable rate tuple $\left( R_1, \ldots, R_K \right)$. Thus, there exist $K$ caching functions ${\phi _{\boldsymbol{\sigma},1}}, \ldots, {\phi _{\boldsymbol{\sigma},K}}$, an encoding function $\psi_{\boldsymbol{\sigma},\textbf{d}}$, and $K$ decoding functions ${\mu _{\textbf{d},1}}, \ldots, {\mu _{\textbf{d},K}}$, which, for large enough $n$, ${P_{e}} < \varepsilon$, where $\varepsilon$ tends to 0 as $n \to \infty$. From Fano's inequality, we have
\begin{align}\label{ProofLemmaBCFano}
R_k - \varepsilon \le \frac{1}{n} I\left( \bigcup\limits_{l=1}^{k} W_{d_k}^{(l)};Y_{\textbf{d},k}^n,U_k \right), \quad \mbox{for $k \in [K]$}.
\end{align}
Accordingly,
\begin{subequations}
\label{ProofLemmaBCFanoR1}
\begin{align}\label{ProofLemmaBCFanoR11}
R_1 - {\varepsilon} \le & \frac{1}{n} I\left( W_{d_1}^{(1)};Y_{\textbf{d},1}^n,U_1 \right)\\
\label{ProofLemmaBCFanoR12}
= & \frac{1}{n} I\left( W_{d_1}^{(1)};U_1 \right) + \frac{1}{n} I\left( W_{d_1}^{(1)};Y_{\textbf{d},1}^n \left| U_1 \right. \right),
\end{align}
\end{subequations}
where the second term in \eqref{ProofLemmaBCFanoR12} can be bounded as follows:
\begin{subequations}
\label{ProofLemmaBCFanoR1SecondTerm}
\begin{align}\label{ProofLemmaBCFanoR1SecondTerm1}
\frac{1}{n} I\left( W_{d_1}^{(1)};Y_{\textbf{d},1}^n \left| U_1 \right. \right) = & \frac{1}{n} \sum\limits_{i=1}^{n}  I\left( W_{d_1}^{(1)};Y_{\textbf{d},1,i} \left| U_1,Y_{\textbf{d},1}^{i-1} \right. \right)\\
\label{ProofLemmaBCFanoR1SecondTerm2}
\le & \frac{1}{n} \sum\limits_{i=1}^{n}  I\left( W_{d_1}^{(1)},Y_{\textbf{d},1}^{i-1};Y_{\textbf{d},1,i} \left| U_1 \right. \right),
\end{align}
\end{subequations}
where we define $Y_{\textbf{d},k}^{i} \buildrel \Delta \over = \left( Y_{\textbf{d},k,1}, \ldots, Y_{\textbf{d},k,i} \right)$, for $k \in [K]$ and $i \in [n]$. Let $T$ be a random variable uniformly distributed over $[n]$ and independent from all other random variables. We have
\begin{subequations}
\label{ProofLemmaBCFanoR1SecondTermSec}
\begin{align}\label{ProofLemmaBCFanoR1SecondTermSec1}
\frac{1}{n} \sum\limits_{i=1}^{n}  I\left( W_{d_1}^{(1)},Y_{\textbf{d},1}^{i-1};Y_{\textbf{d},1,i} \left| U_1 \right. \right) = & I\left( W_{d_1}^{(1)},Y_{\textbf{d},1}^{T-1};Y_{\textbf{d},1,T} \left| U_1,T \right. \right)\\
\label{ProofLemmaBCFanoR1SecondTermSec2}
\le & I\left( W_{d_1}^{(1)},Y_{\textbf{d},1}^{T-1},U_1,T;Y_{\textbf{d},1,T} \right)\\
= & I\left( V_{\textbf{d},1};Y_{\textbf{d},1} \right),
\end{align}
\end{subequations}
where we define $V_{\textbf{d},1} \buildrel \Delta \over = \left( W_{d_1}^{(1)},Y_{\textbf{d},1}^{T-1},U_1,T \right)$, and $Y_{\textbf{d},1} \buildrel \Delta \over = \left( Y_{\textbf{d},1,T} \right)$. From \cref{ProofLemmaBCFanoR1,ProofLemmaBCFanoR1SecondTerm,ProofLemmaBCFanoR1SecondTermSec}, \eqref{ShirinAppendixLowerBound1} is proved. We also have, for $k \in [2:K]$,  
\begin{subequations}
\label{ProofLemmaBCFanoRk}
\begin{align}\label{ProofLemmaBCFanoRk1}
R_k - {\varepsilon} \le & \frac{1}{n} I\left( \bigcup\limits_{l=1}^{k} W_{d_k}^{(l)};Y_{\textbf{d},k}^n,U_k \right)\\
\label{ProofLemmaBCFanoRk2}
\le & \frac{1}{n} I\left( \bigcup\limits_{l=1}^{k} W_{d_k}^{(l)};Y_{\textbf{d},k}^n,U_k \left| \bigcup\limits_{m=1}^{k-1} \bigcup\limits_{l=1}^{m} W_{d_m}^{(l)} \right. \right)\\
\label{ProofLemmaBCFanoRk3}
\le & \frac{1}{n} I\left( \bigcup\limits_{l=1}^{k} W_{d_k}^{(l)};Y_{\textbf{d},k}^n,U_1,\ldots,U_k \left| \bigcup\limits_{m=1}^{k-1} \bigcup\limits_{l=1}^{m} W_{d_m}^{(l)} \right. \right)\\
\label{ProofLemmaBCFanoRk4}
= & \frac{1}{n} I\left( \bigcup\limits_{l=1}^{k} W_{d_k}^{(l)};U_1,\ldots,U_k \left| \bigcup\limits_{m=1}^{k-1} \bigcup\limits_{l=1}^{m} W_{d_m}^{(l)} \right. \right) +\nonumber\\
& \frac{1}{n} I\left( \bigcup\limits_{l=1}^{k} W_{d_k}^{(l)}; Y_{\textbf{d},k}^n \left| U_1,\ldots,U_k, \bigcup\limits_{m=1}^{k-1} \bigcup\limits_{l=1}^{m} W_{d_m}^{(l)} \right. \right),
\end{align}
\end{subequations}
where \eqref{ProofLemmaBCFanoRk2} follows from the independence of the files. We now bound the second term in \eqref{ProofLemmaBCFanoRk4} as follows:
\begin{subequations}
\label{ProofLemmaBCFanoRkSecondTerm}
\begin{align}\label{ProofLemmaBCFanoRkSecondTerm1}
\frac{1}{n} &I\left( \bigcup\limits_{l=1}^{k} W_{d_k}^{(l)}; Y_{\textbf{d},k}^n \left| U_1,\ldots,U_k, \bigcup\limits_{m=1}^{k-1} \bigcup\limits_{l=1}^{m} W_{d_m}^{(l)} \right. \right) \nonumber\\
& = \frac{1}{n} \sum\limits_{i=1}^{n}  I\left( \bigcup\limits_{l=1}^{k} W_{d_k}^{(l)}; Y_{\textbf{d},k,i} \left| U_1,\ldots,U_k, \bigcup\limits_{m=1}^{k-1} \bigcup\limits_{l=1}^{m} W_{d_m}^{(l)}, Y_{\textbf{d},k}^{i-1} \right. \right)\\
\label{ProofLemmaBCFanoRkSecondTerm2}
& \le \frac{1}{n} \sum\limits_{i=1}^{n}  I\left( \bigcup\limits_{l=1}^{k} W_{d_k}^{(l)}; Y_{\textbf{d},k,i} \left| U_1,\ldots,U_k, \bigcup\limits_{m=1}^{k-1} \bigcup\limits_{l=1}^{m} W_{d_m}^{(l)}, Y_{\textbf{d},k}^{i-1}, Y_{\textbf{d},k-1}^{i-1}, \ldots, Y_{\textbf{d},1}^{i-1} \right. \right)\\
\label{ProofLemmaBCFanoRkSecondTerm3}
& \le \frac{1}{n} \sum\limits_{i=1}^{n}  I\left( \bigcup\limits_{l=1}^{k} W_{d_k}^{(l)}, Y_{\textbf{d},k}^{i-1}; Y_{\textbf{d},k,i} \left| U_1,\ldots,U_k, \bigcup\limits_{m=1}^{k-1} \bigcup\limits_{l=1}^{m} W_{d_m}^{(l)}, Y_{\textbf{d},k-1}^{i-1}, \ldots, Y_{\textbf{d},1}^{i-1} \right. \right)\\
\label{ProofLemmaBCFanoRkSecondTerm4}
& = I\left( \bigcup\limits_{l=1}^{k} W_{d_k}^{(l)}, Y_{\textbf{d},k}^{T-1}; Y_{\textbf{d},k,T} \left| U_1,\ldots,U_k, \bigcup\limits_{m=1}^{k-1} \bigcup\limits_{l=1}^{m} W_{d_m}^{(l)}, Y_{\textbf{d},k-1}^{T-1}, \ldots, Y_{\textbf{d},1}^{T-1},T \right. \right)\\
\label{ProofLemmaBCFanoRkSecondTerm5}
& \le I\left( \bigcup\limits_{l=1}^{k} W_{d_k}^{(l)}, Y_{\textbf{d},k}^{T-1},U_k; Y_{\textbf{d},k,T} \left| U_1,\ldots,U_{k-1}, \bigcup\limits_{m=1}^{k-1} \bigcup\limits_{l=1}^{m} W_{d_m}^{(l)}, Y_{\textbf{d},k-1}^{T-1}, \ldots, Y_{\textbf{d},1}^{T-1},T \right. \right)\\
\label{ProofLemmaBCFanoRkSecondTerm6}
& = I\left( {V}_{\textbf{d},k}; Y_{\textbf{d},k} \left| {V}_{\textbf{d},k-1} \right. \right),
\end{align}
\end{subequations}
where $V_{\textbf{d},k} \buildrel \Delta \over = \left( V_{\textbf{d},k-1},\bigcup\limits_{l=1}^{k} W_{d_k}^{(l)}, Y_{\textbf{d},k}^{T-1},U_k \right)$, and $Y_{\textbf{d},k} \buildrel \Delta \over = \left( Y_{\textbf{d},k,T} \right)$, for $k \in [2:K]$. We also note that $V_{\textbf{d},K} = X_{\textbf{d}}$. By plugging \eqref{ProofLemmaBCFanoRkSecondTerm} into \eqref{ProofLemmaBCFanoRk}, the proof of Lemma \ref{LemmaBCProofOuterBound} is completed.

\bibliographystyle{IEEEtran}
\bibliography{Report}

\end{document}